\documentclass[pre,onecolumn,showpacs,preprintnumbers,pra,showkeys,
superscriptaddress]{revtex4}
\usepackage{amsfonts}
\usepackage{amsmath}
\usepackage{graphicx}
\usepackage{dcolumn}
\usepackage{epsfig}
\usepackage{color}
\usepackage{hyperref}
\usepackage{hyphenat}
\usepackage{bm}
\usepackage{float}
\usepackage{amsthm}
\usepackage{enumitem}
\usepackage[paperwidth=210mm,paperheight=297mm,centering,hmargin=2cm,vmargin=2.5cm]{geometry}

\newcommand{\be}{\begin{equation}}
\newcommand{\beast}{\begin{equation*}}
\newcommand{\ee}{\end{equation}}
\newcommand{\eeast}{\end{equation*}}
\newcommand{\br}{\begin{eqnarray}}
\newcommand{\brast}{\begin{eqnarray*}}
\newcommand{\er}{\end{eqnarray}}
\newcommand{\erast}{\end{eqnarray*}}
\newcommand{\bse}{\begin{subequations}}
\newcommand{\ese}{\end{subequations}}

\newcommand{\bd}{\begin{displaymath}}
\newcommand{\ed}{\end{displaymath}}

\newcommand{\bfig}{\begin{figure}}
\newcommand{\efig}{\end{figure}}

\newcommand{\aspe}{\textquotedblleft}
\newcommand{\aspd}{\textquotedblright}
\newcommand{\modu}{\, \mathrm{mod}\,}
\newtheorem{thm}{Theorem}
\newtheorem{lemma}{Lemma}
\newtheorem{expl}{Example}
\newtheorem{algm}{Algorithm}

\begin{document}

\title{An algorithm for hiding and recovering data using matrices}
\author{$^1$Salomon S. Mizrahi} \author{$^2$Di\'{o}genes Galetti}
\address{$^1$Departamento de F\'{i}sica, CCET, Universidade Federal de São Carlos, 
São Carlos, SP, Brasil} \email{salomon@df.ufscar.br} 
\address{$^2$Instituto de F\'{\i}sica Te\'{o}rica, Universidade Estadual Paulista (UNESP), 
S\~{a}o Paulo, SP, Brasil } \email{galetti@ift.unesp.br}
\date{\today}

\begin{abstract}
We present an algorithm for the recovery of a matrix $\mathbb{M}$ %
(non-singular $\in $ $\mathbb{C}^{N\times N}$) by only being aware of two 
of its powers, $\mathbb{M}_{k_{1}}:=\mathbb{M}^{k_{1}}$ and $\mathbb{M}%
_{k_{2}}:=\mathbb{M}^{k_{2}}$ ($k_{1}>k_{2}$) whose exponents are
positive coprime numbers. The knowledge of the exponents is the key to retrieve
matrix $\mathbb{M}$ out from the two matrices $\mathbb{M}_{k_{i}}$. The 
procedure combines products and inversions of matrices, and a few computational 
steps are needed to get $\mathbb{M}$, almost independently of the exponents 
magnitudes. Guessing the matrix $\mathbb{M}$ from the two matrices $\mathbb{M}_{k_{i}}$,
without the knowledge of $k_{1}$ and $k_{2}$, is comparatively 
highly consuming in terms of number of operations. If a private message, 
contained in $\mathbb{M}$, has to be conveyed, the exponents can be encrypted 
and then distributed through a public key method as, for instance, the DF 
(Diffie-Hellman), the RSA (Rivest-Shamir-Adleman), or any other.  
\end{abstract}
\maketitle
%
\section{Introduction}
%
The sender of a message $M$ can choose to convey it to a receiver as a single continuous 
string of numbers, in a chosen base as, binary, decimal, hexadecimal, etc... or he 
also can choose to use the matrix architecture, where each entry conceals part of the message. 
He can also split a string, or a matrix, in several parts to be sent separately. 
Besides being easy to handle, the structure and properties of matrices allow to 
accommodate not only one but several messages if one considers that each entry, 
or each row (or column), corresponds to an instruction or a sentence. Now,If the 
message must be private, or confidential, then, as it was already practiced in the 
antiquity (for instance, by Julius Ceasar) the message must be encrypted and hopping 
that only the receiver could get access and knows the procedure to decrypt it.  
To our knowledge the use of matrices to encode messages was 
initiated by L. S. Hill \cite{hill}. The Hill cipher is a polygraphic substitution 
that makes use of matrix multiplications in order to change a plaintext letters into 
a ciphertext. One of the basic components of classical ciphers is the substitution 
cipher: a ciphertext matrix $\mathbb{C}$ is obtained by multiplication of a
plaintext matrix $\mathbb{P}$ by a key matrix $\mathbb{K}$, $\mathbb{C=KP}%
\modu \left( q\right)$, i.e., by a linear transformation to be followed
by an arithmetic modular operation, $q$ being the number of digits. So a square 
matrix $\mathbb{C}$, of order $N$ can host, for instance, $2N^{2}$ messages to 
be conveyed by the sender $A$ to the receiver $B$. More recently, a report  
about the usefulness of matrices in public-key cryptography was published 
\cite{ayan1}.

When a sender $A$ (aka Alice, in the jargon) wants to remit a message
contained in a matrix $\mathbb{A}$ to a receiver $B$ (aka Bob), such that it could
not be known by any third person $E$ (aka the eavesdropper Eve, that is also a 
cryptanalyst), it is necessary to encode $\mathbb{A}\Longrightarrow \mathbb{M}$ 
in some special form such to be hard to be decoded by Eve, conceding that she 
has access to $\mathbb{M}$. That this matrix is difficult to decrypt means
that even if Eve succeeds to decode $\mathbb{M}$ and gets the
message, the time she consumed for the task is long enough so that the acquired
data becomes already obsolete. The standard procedure to encrypt $%
\mathbb{A}$ into $\mathbb{M}$ consists in producing of a secret key
owned only by Bob that he uses to get $\mathbb{M}$ in a very short time
compared to the one Eve needs to arrive at the same result. 

If higher confidentiality of the message is needed, then a 
second layer of encryption can be implemented. For instance, a matrix $\mathbb{M%
}$, already subjected to a first encryption, is subjected to a second one,
becoming matrix $\mathbb{Z}$. An interesting procedure could be the
calculation of $\mathbb{Z :=M}^{k}$ the $k$-th ($k$ an integer) power of a
non-singular square matrix defined in the field of complex numbers $\mathbb{M%
}\in C^{N\times N}$. The inverse process to extract the unkown $\mathbb{M}$
from a know $\mathbb{Z}$ consists in seeking a solution $\mathbb{X}$ (or
more than one) of the equation $\mathbb{X}=\mathbb{Z}^{1/k}$ and expect that 
$\mathbb{X=M}$. The direct procedure consists in diagonalizing $\mathbb{Z}$
and to each eigenvalue $z_{i}$ solve the equation $x^{k}-z_{i}=0$ and the zeros $%
x_{i,l},$ $l=1,...,k$ will result; any zero can be one eigenvalue of
matrix $\mathbb{M}$. Here lies the ambiguity! An eavesdropper may have
access to the matrix $\mathbb{Z}$ but Alice and Bob have to keep $k$ as they
secret key. Several algorithms are proposed in Ref. \cite{higham1} that
make use of iterative approaches, which however lead to approximate solutions.

Nevertheless, one could envisage another method still based on powers of a
matrix $\mathbb{M}$, making the unraveling of $\mathbb{Z}$ less time
consuming than to construct it, since $k-1$ matrix multiplications are
necessary to construct $\mathbb{M}^{k}$, without storage of intermediary results. 
We propose an algorithm to reconstruct a matrix $\mathbb{M}\in C^{N\times N}$, 
to be conveyed from Alice to Bob, having arbitrary eigenvalue multiplicity 
(degeneracy). The only requirement is that $\mathbb{M}$ be non-singular
\footnote{Even being singular, the singularity may be removed by the 
addition spurious row and column.}. The procedure goes as follows: Alice 
compute two powers of $\mathbb{M}$, with 
exponents that are two positive coprime (relatively prime) numbers, $k_{1}$ 
and $k_{2}$, such that $\mathbb{M}_{k_{1}}:=\mathbb{M}^{k_{1}}$ and 
$\mathbb{M}_{k_{2}}:=\mathbb{M}^{k_{2}}$, and sends the matrices $\mathbb{M}_{k_{1}}$ 
and $ \mathbb{M}_{k_{2}}$ to Bob, publicly or privately. Bob's task consists in 
recovering $\mathbb{M}_{1}:=\mathbb{M}$ from those two matrices in a quite
short time (the minimal number of matrix operations), before it could 
be unraveled by a third party. The knowledge of the exponents, $k_{1}$ and $k_{2}$, 
is the key element to invert the process in order to get $\mathbb{M}_{1}$. 
The exponents can be either shared by Alice and Bob previously, or they can be 
encoded, creating an \emph{encyphering key}, by using, for instance, the 
Diffie-Hellman proposal \cite{diffie}, the RSA method \cite{RSA} or the 
quantum BB84 one \cite{BB84}, and then distributed by 
Bob (indeed, it is Bob that could determine the values $k_{1}$ and $k_{2}$ that Alice
should use), one key can be made public (a unique key for many several senders) 
or private (one key for each message sender) while the other one is Bob's own key 
(or keys) that he keeps secretly.

Regarding the computational time consumption, there are efficient algorithms
that reduce the necessary time to calculate the product of two square $N \times N$
matrices to $O\left( N^{2.376}\right)$ \cite{copper}, instead of the direct method 
that is $O\left( N^{3}\right) $. The inversion of an invertible matrix consumes the 
same amount of time  \cite{copper}. For Alice, the time consumed to calculate the powers
of a matrix $\mathbb{M}^{k}$ can be reduced from $\left( k-1\right) $ to 
approximately $2\sqrt{k}$ for $k\gg 1$, if storage of lower powers ($<k/2$)
matrices is possible. Regarding our algorithm the inverse operation,
extracting $\mathbb{M}$ from the matrices $\mathbb{M}_{k_{1}}$ and $\mathbb{M%
}_{k_{2}}$, turns out to be much less time consuming (to Bob) than to produce
it (by Alice). In the following sections we present the algorithm that can
be used to encrypt messages to be sent from Alice to Bob and still keeping a
degree of confidentiality against a skilful eavesdropper. To illustrate the
procedures of Alice and Bob we present a simple example.

\begin{expl}\label{ex1}
Alice wants to send a message to Bob that is encoded in the
non-singular matrix $\mathbb{M}\in $ $\mathbb{C}^{2\times 2}$. Alice 
computes two powers of $\mathbb{M}$,
choosing, for instance, the exponents $k_{1}=17$ and $k_{2}=11$ (or receives
them from Bob), and she constructs the matrices $\mathbb{C}:=\mathbb{M}_{17}$
and $\mathbb{D}:=\mathbb{M}_{11}$ ,
\begin{equation}
\mathbb{C}=\left( 
\begin{array}{cc}
-8229\,303\,833+i\,96\, 078\,379\,931 & 159\, 011\,972\,369 + i\,193\, 528\,618\,710  \\ 
221\, 214\,366\,306+ i\,117\, 483\,350\,975 &  
600\, 042\,299\,893- i\,104\, 462\,213\,832%
\end{array}%
\right) .  
\label{a1}
\end{equation}
and 
\begin{equation}
\mathbb{D}=\left( 
\begin{array}{cc}
-717\,135-i\, 7281\,379 & -17\,385\,865-i\,6190\,286 \\ 
-18\,429\,370+i\,972\,761 & -31\,527\,077+i\,28\,463\,112%
\end{array}%
\right)   
\label{a2}
\end{equation}
respectively, and sends them to Bob, through some channel, openly or
privately. Alice then produces an enciphering key, $K_{E}$, for the numbers $%
k_{1}$ and $k_{2}$ (or she does not need to send anything else if she
received $k_{1}\ $and $k_{2}$ from Bob) that she sends to Bob through
another channel. Bob uses the deciphering key $K_{D}$, that only he knows,
to retrieve $k_{1}$ and $k_{2}$, and he then makes use of the algorithm (to
be explained below), that consists in doing the matrix operations
$\mathbb{C}^{2}\mathbb{D}^{-3}$ that results in  
\begin{equation}
\mathbb{M}=\left( 
\begin{array}{cc}
1+4i & 3-2i \\ 
2-3i & -1-5i%
\end{array}%
\right)\ .
\label{a3}
\end{equation}
\end{expl}
The algorithm dispenses the need of calculating eigenvectors, eigenvalues or
to perform any decomposition of $\mathbb{C}$ and $\mathbb{D}$ and do not
depend on the possible eigenvalue multiplicities of matrices (\ref{a1}) and (%
\ref{a2}). Once the matrices $\mathbb{C}$ and $\mathbb{D}$ are stored, the
number of matrix operations Bob needs to perform are $4$: one matrix
inversion $\left( \mathbb{D}^{-1}\right) $ (stored), one product $\mathbb{CD}%
^{-1}$ (stored), two products $\left( \mathbb{CD}^{-1}\right) \left( \mathbb{%
CD}^{-1}\right) \left( \mathbb{D}^{-1}\right) $, whereas in order to
construct $\mathbb{C}$ and $\mathbb{D}$ Alice has to calculate $16$ matrix
multiplications without storage of powers of $\mathbb{M}$ or $7$
multiplications considering storage, as to be explained in details below. For 
a different pair of numbers, $k_{1}=38$ and $k_{2}=23$, for instance, as exponents, 
the sequence of operations needed to retrieve $\mathbb{M}$ is a different 
one, being $\mathbb{C}^{-3}\mathbb{D}^{5}$.

The pair of exponents (-3,5) is unique for each pair of coprime integers 
$k_{1}\ $and $k_{2}$, as shown in Theorem \ref{theor3} in Appendix A. As 
observed by D. Knuth \cite{knuth1}, working numerically with integers or 
with fractions, as entries of 
matrix $\mathbb{M}$, instead of using floating point numbers, avoids 
accumulated rounding errors, so the accuracy of the calculations is absolute. 
In the case where an error occurs during the transmission of the key conveyed 
by Alice to Bob (or vice versa), for instance instead of receiving $%
k_{2}=11$ Bob receives the number $12$, then, according to the proposed 
Algorithm \ref{alg1} below, he will calculate a different sequence of
operations involving matrices $\mathbb{C}$ and $\mathbb{D}$ , namely $%
\mathbb{C}^{5} \mathbb{D}^{-7}$, and he will get a fully different matrix,
\begin{equation}
\mathbb{M}^{\prime }=\left( 
\begin{array}{cc}
5\,319 + i\, 84\, 141 & -1\, 142+ i\, 119\,291 \\ 
44\,827+ i\, 110\, 554 & 288\, 728+ i\, 196\, 979%
\end{array}%
\right) ,  \label{a4}
\end{equation}
instead of the sought matrix (\ref{a3}). Therefore, the sequence of operations 
to get the matrix $\mathbb{M}$, from matrices $\mathbb{C}$ and $\mathbb{D}$, is
unique for each pair of coprime numbers $k_{1}$ and $k_{2}$. If an error
occurs by changing $k_{1}$ or $k_{2}$, even by one unit, the sequence of
operations changes, resulting in a wrong matrix as in Eq. (\ref{a4}). 
%
\section{Algorithm and examples}
%
We consider a variant of Euclid's algorithm, see Lemma \ref{lema1} 
(in Appendix A), namely, instead of looking for the $\gcd$ (greatest common 
divisor) of a pair of arbitrary coprime integers -- gcd$(k_{1}$, $k_{2})= 1$ -- 
we are essentially interested in determining the sequence of
quotients that are inherent to the algorithm. For $k_{1}$ and $k_{2}$ 
positive coprime numbers we write down the sequence of modular calculations
of the pairs $\left( k_{r},k_{r-1}\right) $, $k_{1}>k_{2}>\cdots >k_{r-1}>k_{r}>k_{r+1}=1$, 
\begin{eqnarray}
k_{1}\modu\ k_{2} &=& k_{3}\quad \mathrm{\Longrightarrow \quad }%
k_{1}=q_{2}k_{2}+k_{3},  \notag \\
k_{2}\modu\ k_{3} &=&k_{4}\quad \mathrm{\Longrightarrow \quad }%
k_{2}=q_{3}k_{3}+k_{4},  \notag \\
k_{3}\modu\ k_{4} &=&k_{5}\quad \mathrm{\Longrightarrow \quad }%
k_{3}=q_{4}k_{4}+k_{5},  \notag \\
&&\vdots \quad \ \ \   \label{B4} \\
k_{r-2}\modu\ k_{r-1} &=&k_{r}\quad \ \mathrm{\Longrightarrow \quad }%
k_{r-2}=q_{r-1}k_{r-1}+k_{r},  \notag \\
k_{r-1}\modu\ k_{r} &=&k_{r+1}\quad \ \ \mathrm{\Longrightarrow \quad }%
k_{r-1}=q_{r}k_{r}+k_{r+1},  \notag
\end{eqnarray}
where $q_{2},q_{3},...,q_{r}$ are the quotients and $k_{3},k_{4},...k_{r+1}$
the remainders. 

We show next that the determination of a non-singular
matrix, $\mathbb{M\in C}^{N\times N}$, admitted to be unknown, is possible
if: (a) two of its powers, $\mathbb{M}_{k_{1}}:= \mathbb{M} ^{k_{1}}$ and 
$\mathbb{M}_{k_{2}}:= \mathbb{M} ^{k_{2}}$, are known, and (b) the exponents 
$k_{1}$ and $k_{2}$, positive coprime numbers, are also known.
\begin{algm} \label{alg1}
Since $k_{1}=q_{2}k_{2}+k_{3}$ we write the decomposition $%
\mathbb{M}_{k_{3}}=$ $\mathbb{M}_{k_{1}}\left( \left( \mathbb{M}%
_{k_{2}}\right) ^{q_{2}}\right) ^{-1}$ to get the matrix $\mathbb{M}_{k_{3}}$%
;$\,$ again, the numbers ($k_{2}$, $k_{3}$) are also coprime which permits
us to write the decomposition $\mathbb{M}_{k_{4}}=$ $\mathbb{M}%
_{k_{2}}\left( \left( \mathbb{M}_{k_{3}}\right) ^{q_{3}}\right) ^{-1}=%
\mathbb{M}_{k_{2}}\left( \left( \mathbb{M}_{k_{3}}\right) ^{-1}\right)
^{q_{3}}$ to get the matrix $\mathbb{M}_{k_{4}}$. Continuing this procedure
by reducing the powers of matrix $\mathbb{M}_{k_{1}}$, one arrives at the
original matrix $\mathbb{M}_{1}$ (the seed), $\mathbb{M}_{1}=$ $\mathbb{M}%
_{k_{r-1}}\left( \left( \mathbb{M}_{k_{r}}\right) ^{q_{r}}\right) ^{-1}=%
\mathbb{M}_{k_{r-1}}\left( \left( \mathbb{M}_{k_{r}}\right) ^{-1}\right)
^{q_{r}}$. All the $r-1$ pairs $\left( k_{1},k_{2}\right) $, $\left(
k_{2},k_{3}\right) $, ..., $\left( k_{r-1},k_{r}\right) $ are coprime numbers. The
choice by Alice (or by Bob) of $\left( k_{1},k_{2}\right) $ being coprime is
an important feature in order that the last equation of the sequence (\ref%
{B4}) be $k_{r-1}\modu k_{r}=1$; the value $k_{r+1}=1$ is necessary
to retrieve the seed matrix $\mathbb{M}$. The set of quotients $\left(
 q_{2},q_{3},...,q_{r} \right) $ constitutes the essential element to
construct the sequence of operations: 
\begin{eqnarray}
\mathbb{M}_{k_{3}} &=&\mathbb{M}_{k_{1}}\left( \left( \mathbb{M}%
_{k_{2}}\right) ^{-1}\right) ^{q_{2}},  \notag \\
\mathbb{M}_{k_{4}} &=&\mathbb{M}_{k_{2}}\left( \left( \mathbb{M}%
_{k_{3}}\right) ^{-1}\right) ^{q_{3}},  \notag \\
&&\vdots   \label{B5} \\
\mathbb{M}_{k_{r+1}} &=&\mathbb{M}_{1}=\mathbb{M}_{k_{r-1}}\left( \left( 
\mathbb{M}_{k_{r}}\right) ^{-1}\right) ^{q_{r}}.  \notag
\end{eqnarray}
\end{algm}

From another point of view, as to be shown below, any sequence of integers $%
\left( q_{2},q_{3},...,q_{r},q_{r+1}=k_{r}\right) $ expressed as a \emph{%
continued fraction} leads to the ratio $k_{1}/k_{2}$ where $k_{1}$ and $%
k_{2} $ are coprime. Because of the one-to-one correspondence, Alice (or Bob) can
choose either a pair $\left( k_{1}, k_{2} \right)$ or a sequence of quotients 
$\left(  q_{2},q_{3},...,q_{r} \right) $ to
construct the matrices $\mathbb{C}$ and $\mathbb{D}$. However, for not
facilitating the task of decryption by Eve, it is more convenient to choose
judiciously the sequence of integers that will produce the pair of
exponents. The number of operations in the form of products of matrices
(without storage) and inversions is 
\begin{equation}
\mathcal{N}_{op,r}=\sum_{i=2}^{r}\left( q_{i}+1\right) .  \label{b6}
\end{equation}
For any pair of coprime numbers $(k_1,k_2)$ the uniqueness of the 
sequence of quotients is proved in Theorem \ref{theor3}. Below we illustrate 
the use of the Algorithm \ref{alg1} by two examples
\begin{expl}\label{ex2}
For the prime numbers $k_{1}=1019$ and $k_{2}=239$, the sequence
of the modular calculation goes as shown in Table \ref{b7}, 
\begin{center}
\begin{table}[H]
\centering
\begin{equation}
\begin{tabular}{|c||c|l|l|c|}
\hline
$i+1$ & $k_{i+1}$ & $k_{i}\modu \ k_{i+1}= k_{i+2}$ & $%
q_{i+1}k_{i+1}+k_{i+2}=k_{i}$ & $q_{i+1}$ \\ \hline\hline
$2$ & $239$ & $1019\modu \ 239=63$ & $4\times 239+63=1019$ & $4$ \\ 
\hline
$3$ & $63$ & $239\modu \ 63=50$ & $3\times 63+50=239$ & $3$ \\ \hline
$4$ & $50$ & $63\modu \ 50=13$ & $1\times 50+13=63$ & $1$ \\ \hline
$5$ & $13$ & $50\modu \ 13=11$ & $3\times 13+11=50$ & $3$ \\ \hline
$6$ & $11$ & $13\modu \ 11=2$ & $1\times 11+2=13$ & $1$ \\ \hline
$7$ & $2$ & $11\modu \ 2=1$ & $5\times 2+1=11$ & $5$ \\ \hline
$8$ & $1$ & $2\modu \ 1=0$ & $2\times 1+0=2$ & $2$ \\ \hline
\end{tabular}
\end{equation}%
\caption{{\small {Sequence of the modular calculation with
remainders and quotients}}}
\label{b7}
\end{table}
\end{center}
where $\left( q_{2},q_{3},...,q_{8}\right) =\left( 4,3,1,3,1,5,2\right) $, $%
r=7$, and all the pairs $\left( 1019,239\right) $, $\left(
239,63\right) $, $\left( 63,50\right) $, $\left( 50,13\right) $, $\left(
13,11\right) $, $\left( 11,2\right) $, $\left( 2,1\right) $ are coprime. 
From these results the six sequential operations to be done by Bob are 
\begin{subequations}
\label{b9}
\begin{eqnarray}
\mathbb{M}_{63} &=&\mathbb{M}_{1019}\left( \left( \mathbb{M}_{239}\right)
^{-1}\right) ^{4},(5)  \label{b9a} \\
\mathbb{M}_{50} &=&\mathbb{M}_{239}\left( \left( \mathbb{M}_{63}\right)
^{-1}\right) ^{3},\left( 4\right)   \label{b9b} \\
\mathbb{M}_{13} &=&\mathbb{M}_{63}\left( \mathbb{M}_{50}\right) ^{-1},\text{ 
}\left( 2\right)   \label{b9c} \\
\mathbb{M}_{11} &=&\mathbb{M}_{50}\left( \left( \mathbb{M}_{13}\right)
^{-1}\right) ^{3},\left( 4\right)   \label{b9d} \\
\mathbb{M}_{2} &=&\mathbb{M}_{13}\left( \mathbb{M}_{11}\right) ^{-1},\left(
2\right)   \label{b9f} \\
\mathbb{M}_{1} &=&\mathbb{M}_{11}\left( \left( \mathbb{M}_{2}\right)
^{-1}\right) ^{5},\left( 6\right) .  \label{b9g}
\end{eqnarray}
\end{subequations}
Each step is calculated as $\mathbb{M}_{k_{i+2}}=\mathbb{M}_{k_{i}}\left(
\left( \mathbb{M}_{k_{i+1}}\right) ^{-1}\right) ^{q_{i+1}}$, and in each row, 
in parenthesis, we give the number of operations (matrix multiplications
plus inversions) which is $\left( q_{i}+1\right) $. From Table \ref{b7} we observe 
that the total number of operations calculated from Eq. (\ref{b6}) is $%
\mathcal{N}_{op,7}=23$. Calling $\mathbb{M}_{1019}=\mathbb{C}$ and $\mathbb{M%
}_{239}=\mathbb{D}$, we then rewrite each row in the set of Eqs. (\ref{b9}) as 
\begin{subequations}
\label{b10}
\begin{eqnarray}
\mathbb{M}_{63} &=&
\mathbb{C}\left( \mathbb{D}^{-1}\right) ^{4},\quad 
\label{b10a} \\
\mathbb{M}_{50} &=&
\left( \mathbb{C}^{-1}\right) ^{3}\mathbb{D}^{13},\quad 
\label{b10b} \\
\mathbb{M}_{13} &=&
\mathbb{C}^{4}\left( \mathbb{D}^{-1}\right) ^{17},\quad   \label{b10c} \\
\mathbb{M}_{11} &=&
\left( \mathbb{C}^{-1}\right) ^{15}\mathbb{D}^{64},\quad 
\label{b10d} \\
\mathbb{M}_{2} &=&
\mathbb{C}^{19}\left( \mathbb{D}^{-1}\right) ^{81},\quad   \label{b10f} \\
\mathbb{M}_{1} &=&
\left( \mathbb{C}^{-1}\right) ^{110}\mathbb{D}^{469},\quad 
\label{b10g}
\end{eqnarray}
\end{subequations}
and the necessary number of operations with storage to be performed in the RHS 
of Eq. (\ref{b10g}) is: $6$ 
matrix inversions and $15$ multiplications ($3+3+1+3+1+4=15$). So the total 
number of matrix operations reduces to $21$. 

We anticipate the result of Theorem \ref{teor3} that says: in the RHS of 
each equation in the set (\ref{b10}) the exponents in 
$\mathbb{C}^{p_{l}}\mathbb{D}^{t_{l}}$ are related through the equation
\begin{equation}
p_{l}k_{1}+t_{l}k_{2}=k_{l+2},\quad l=1,..,6\ ,  \label{b11}
\end{equation}
where $p_{l}$ and $t_{l}$ are two integers such that $p_{l}t_{l}<0$, whose
dependence on the quotients is shown in Table \ref{b12}, and the last line, 
$p_{6}k_{1}+t_{6}k_{2}=1$, is known as \emph{Bezout identity}.  
\begin{table}[hbtp]
\begin{tabular}{|l||c|c|c|}
\hline
$l$ & $p_{l}$ & $t_{l}$ & $p_{l}k_{1}+t_{l}k_{2}=k_{l+2}$ \\ \hline\hline
$1$ & \multicolumn{1}{r}{$1$} & \multicolumn{1}{|r|}{$-4$} & 
\multicolumn{1}{r|}{$1\times 1019-4\times 239=63$} \\ \hline
$2$ & \multicolumn{1}{r}{$-3$} & \multicolumn{1}{|r|}{$13$} & 
\multicolumn{1}{r|}{$-3\times 1019+13\times 239=50$} \\ \hline
$3$ & \multicolumn{1}{r}{$4$} & \multicolumn{1}{|r|}{$-17$} & 
\multicolumn{1}{r|}{$4\times 1019-17\times 239=13$} \\ \hline
$4$ & \multicolumn{1}{r}{$-15$} & \multicolumn{1}{|r|}{$64$} & 
\multicolumn{1}{r|}{$-15\times 1019+64\times 239=11$} \\ \hline
$5$ & \multicolumn{1}{r|}{$19$} & $-81$ & \multicolumn{1}{r|}{$19\times
1019-81\times 239=2$} \\ \hline
$6$ & \multicolumn{1}{r|}{$-110$} & $469$ & \multicolumn{1}{r|}{$%
-110\times 1019+469\times 239=1$} \\ \hline
\end{tabular}%
\caption{\small{Bezout relations and the coefficients.}}
\label{b12}
\end{table}
The moduli of the coefficient $\left\vert p_{l}\right\vert $ and $\left\vert
t_{l}\right\vert $ can be used for calculating the number of necessary operations
(matrix products) to retrieve the matrices $\mathbb{M}%
_{k_{i}}$ out of the matrices $\mathbb{C}$ and $\mathbb{D}$. In fact,
instead of sending to Bob the coprime numbers $k_{1}$ and $k_{2}$ Alice
could choose to send the coefficients of the Bezout relation, $p_{6}=-110$
and $t_{6}=$ $469$, that are not necessarily coprime numbers, thus sparing
time for Bob to do any further calculations. However this approach has a
drawback, it asks a lot of memory for storage since huge matrices ($\mathbb{C%
},\mathbb{D}$) have to be exponentiated, $\mathbb{C}^{p_{l}}\mathbb{D}%
^{t_{l}}$, with quite large exponents, whereas considering the sequence of
matrices of the kind $\left( \left( \mathbb{M}_{k_{i}}\right) ^{-1}\right)
^{q_{i}}$ where the exponents (the quotients $q_{2},q_{3},...q_{r},q_{r+1}$)
are quite small compared with the coefficient in the Bezout identity 
-- so sparing the use of large amount of bytes and
computational time -- that turns out to be advantageous to Bob if his
computational capabilities are limited.
\end{expl}
\begin{expl} \label{ex3} 
For the coprime numbers $k_{1}=1001$ and $k_{2}=213$, the
sequence of modular calculations is given in Table \ref{b13} 
\begin{table}[htbp]
\begin{tabular}{|c||c|l|l|c|}
\hline
$i+1$ & $k_{i+1}$ & $k_{i}\modu\ k_{i+1}= k_{i+2}$ & $%
q_{i+1}k_{i+1}+k_{i+2}=k_{i}$ & $q_{i+1}$ \\ \hline\hline
$2$ & $213$ & $1001\modu \ 213=149$ & $4\times 213+149=1001$ & $4$ \\ 
\hline
$3$ & $149$ & $213\modu\ 149=64$ & $1\times 149+64=213$ & $1$ \\ \hline
$4$ & $64$ & $149\modu \ 64=21$ & $2\times 64+21=149$ & $2$ \\ \hline
$5$ & $21$ & $64\modu\ 21=1$ & $3\times 21+1=64$ & $3$ \\ \hline
$6$ & $1$ & $21\modu\ 1=0$ & $21\times 1+0=21$ & $21$ \\ \hline
\end{tabular}
\caption{{\small{Sequence of the modular calculations with remainders and quotients.}}}
\label{b13}
\end{table}
with the quotients $\left( 4,1,2,3,21\right)$. From Eq. (\ref{b6}),   
the number of operations (without storage) is $\mathcal{N}_{op,5}=14$ 
and the sequence of calculation steps is 
\begin{subequations}
\label{b25}
\begin{eqnarray}
\mathbb{M}_{149} &=&\mathbb{M}_{1001}\left( \left( \mathbb{M}_{213}\right)
^{-1}\right) ^{4}=\mathbb{C}\left( \mathbb{D}^{-1}\right) ^{4},  \label{b25a}
\\
\mathbb{M}_{64} &=&\mathbb{M}_{213}\left( \mathbb{M}_{149}\right) ^{-1}=%
\mathbb{C}^{-1}\mathbb{D}^{5},  \label{b25b} \\
\mathbb{M}_{21} &=&\mathbb{M}_{149}\left( \left( \mathbb{M}_{64}\right)
^{-1}\right) ^{2}=\mathbb{C}^{3}\left( \mathbb{D}^{-1}\right) ^{14},
\label{b25c} \\
\mathbb{M}_{1} &=&\mathbb{M}_{64}\left( \left( \mathbb{M}_{21}\right)
^{-1}\right) ^{3}=\left( \mathbb{C}^{-1}\right) ^{10}\mathbb{D}^{47}\ .
\label{b25d}
\end{eqnarray}
\end{subequations}
The relations $p_{l}k_{1}+t_{l}k_{2}=k_{l+2}$ ($p_{l}t_{l}<0$) are
verified in Table \ref{b26}. 
\begin{table}[htbp]
\begin{tabular}{|l||c|c|c|}
\hline
$l$ & $p_{l}$ & $t_{l}$ & $k_{l+2}$ \\ \hline\hline
$1$ & \multicolumn{1}{r}{$1$} & \multicolumn{1}{|r|}{$-4$} & 
\multicolumn{1}{r|}{$149$} \\ \hline
$2$ & \multicolumn{1}{r}{$-1$} & \multicolumn{1}{|r|}{$5$} & 
\multicolumn{1}{r|}{$64$} \\ \hline
$3$ & \multicolumn{1}{r}{$3$} & \multicolumn{1}{|r|}{$-14$} & 
\multicolumn{1}{r|}{$21$} \\ \hline
$4$ & \multicolumn{1}{r}{$-10$} & \multicolumn{1}{|r|}{$47$} & 
\multicolumn{1}{r|}{$1$} \\ \hline
\end{tabular}%
\caption{\small{Coefficients and remainders.}}
\label{b26}
\end{table}
\end{expl}
The necessary times needed to encrypt and to decrypt a message are unbalanced. If
Alice has a high performance computer but Bob does not or he reckons on a
short time to decode the encryption, then the algorithm works well for him
as long as he knows the numbers $k_{1}$ and $k_{2}$, whereas an eavesdropper
should consume a time that can be comparatively much larger to break the
code and get $\mathbb{M}_{1},$ as to be discussed below. In order to
retrieve $\mathbb{M}_{1}$ from, for instance, the matrices of example 
(\ref{ex3}), $\mathbb{C=M}_{1001}$ and $\mathbb{D=M}_{213}$, Bob has to 
perform $13$ operations: the
factor $\left( \mathbb{C}^{-1}\right) ^{10}$ needs $4$ and $\mathbb{D}^{47}$
needs $8$, thus to get $\mathbb{M}_{1}=\left( \mathbb{C}^{-1}\right) ^{10}%
\mathbb{D}^{47}$, $13$ operations with storage of partial products are
necessary. However, to produce the matrices $\mathbb{M}_{1001}$ and $\mathbb{%
M}_{213}$ Alice has to perform $14$ plus $3$ matrix products (with storage)
respectively, so $17$ in total.
%
\section{Theorems: continued fractions and quotients} \label{cfaq}
%
We now explore the relation between the exponent $k_{1}$ and $k_{2}$
and a continued fraction involving the quotients, and show its usefulness
for Alice to choose the numbers $k_{1}$ and $k_{2}$. We make use of a
theorem presented in \cite{knuth2,khinchin}:

\begin{thm}\label{teor2}
There is a one-to-one correspondence between the coprime numbers $k_{1}$ 
and $k_{2}$ and the sequence of quotients that makes the continued fraction 
$\left/ \left/ q_{2},q_{3},...,q_{r},q_{r+1}\right/\right/ $ (it is customary 
to denote the sequence of numbers in a continued fraction between two double 
slashes $\left/ \left/ ...\right/ \right/ $). Alternatively, 
given a sequence of chosen positive integers 
$\left( q_{2},q_{3},...,q_{r},q_{r+1}\right) $, the calculation of the finite
continued fraction, involving the $q_{i}$'s, results in the ratio of coprime
numbers $k_{1}/k_{2}$.
\end{thm}
\begin{proof}
We write the equations that are on the left side in the first and second
lines of the set (\ref{B4}) as
\begin{eqnarray}
\frac{k_{1}}{k_{2}} &=&q_{2}+\frac{k_{3}}{k_{2}},  \label{b4a} \\
\frac{k_{2}}{k_{3}} &=&q_{3}+\frac{k_{4}}{k_{3}},  \label{b4b}
\end{eqnarray}
we then substitute $\left( k_{2}/k_{3}\right) ^{-1}$ from Eq. (\ref{b4b})
into Eq. (\ref{b4a}) to obtain 
\begin{equation}
\frac{k_{1}}{k_{2}}=q_{2}+\cfrac{1}{q_{3}+\cfrac{1}{k_{3}/k_{4}}}\ .  
\label{b4c}
\end{equation}
One repeats this operation for the finite continued fraction for all the 
$q_{i}$'s and the result is the ratio 
\begin{equation}
\frac{k_{1}}{k_{2}}=q_{2}+\cfrac{1}{q_{3}+\cfrac{1}{q_{4}+\cfrac{1}{\ddots +%
\cfrac{1}{q_{r+1}}}}}:=\left/ \left/ q_{2},q_{3},...,q_{r},q_{r+1}\right/
\right/  \label{b4d}
\end{equation}
where $q_{r+1}=k_{r}$, because $k_{r}=q_{r+1}k_{r+1}+0=q_{r+1}\times 1+0$.
\end{proof}
\begin{expl} \label{ex4} 
For the coprime numbers $\left( k_{1},k_{2}\right) =\left(
1019,239\right) $, the sequence of modular calculations in Table \ref{b7} permits
us to write the finite continued fraction (\ref{b4d}) as 
\begin{equation}
\left/ \left/ 4,3,1,3,1,5,2\right/ \right/ =4+\cfrac{1}{3+\cfrac{1}{1+\cfrac{1}{%
3+\cfrac{1}{1+\cfrac{1}{5+\cfrac{1}{2}}}}}}=\cfrac{1019}{239}\ ,  \label{b4f}
\end{equation}
which is the ratio $k_{1}/k_{2}$. See also \cite{khinchin}
\end{expl}

Another theorem related to Theorem \ref{teor2} and useful for our method follows:
\begin{thm} \label{teor3} 
A one-to-one correspondence between the coprime numbers $%
\left( k_{1},k_{2}\right) $ and the sequence of quotients $\left( 
q_{2},q_{3},...,q_{r},q_{r+1}\right)  $ of Eq. (\ref{B4}) exists and is given as
the inverse of the product of a finite number of $2\times 2$ matrices, 
\begin{equation}
\mathbb{B}_{r} := \left[ \prod_{l=2}^{r}\left( 
\begin{array}{cc}
q_{l} & 1 \\ 
1 & 0%
\end{array}%
\right) \right]^{-1},  
\label{bez}
\end{equation}
where 
\begin{equation} 
\mathbb{B}_{r} \left( 
\begin{array}{c}
k_{1} \\ 
k_{2}%
\end{array}%
\right) =\left( 
\begin{array}{c}
k_{r} \\ 
1%
\end{array}%
\right),  \label{b4g}
\end{equation}
and reminding that $k_{r}=q_{r+1}$.
\end{thm}
\begin{proof}
We consider the equations on the right side of the set (\ref{B4}) 
and write the first two equations as 
\begin{equation}
\left( 
\begin{array}{c}
k_{1} \\ 
k_{2}%
\end{array}%
\right) =\left( 
\begin{array}{cc}
q_{2} & 1 \\ 
1 & 0%
\end{array}%
\right) \left( 
\begin{array}{c}
k_{2} \\ 
k_{3}%
\end{array}%
\right)\ ;  \label{b4h}
\end{equation}
proceeding in the same fashion for the second and third equations,
\begin{equation}
\left( 
\begin{array}{c}
k_{2} \\ 
k_{3}%
\end{array}%
\right) =\left( 
\begin{array}{cc}
q_{3} & 1 \\ 
1 & 0%
\end{array}%
\right) \left( 
\begin{array}{c}
k_{3} \\ 
k_{4}%
\end{array}%
\right)  \label{b4i}
\end{equation}
and substituting Eq. (\ref{b4i}) into Eq. (\ref{b4h}) we get
\begin{equation}
\left( 
\begin{array}{c}
k_{1} \\ 
k_{2}%
\end{array}%
\right) =\left( 
\begin{array}{cc}
q_{2} & 1 \\ 
1 & 0%
\end{array}%
\right) \left( 
\begin{array}{cc}
q_{3} & 1 \\ 
1 & 0%
\end{array}%
\right) \left( 
\begin{array}{c}
k_{3} \\ 
k_{4}%
\end{array}%
\right) .  \label{b4j}
\end{equation}
By repeating the process of substitution, with the last $k_{r+1}=1$, and then 
inverting the relation one gets Eq. (\ref{b4g}).
\end{proof}

\begin{expl} \label{ex5} 
For the prime numbers $k_{1}= 1019$ and $k_{2}=239$, the sequence
of quotients is $\left( q_{2},q_{3},...,q_{8}\right) =\left(
4,3,1,3,1,5,2\right) $, with $r=7$ and $k_{7}=2$; so 
\begin{eqnarray}
\left( 
\begin{array}{c}
1019 \\ 
239%
\end{array}%
\right)  &=&\left[ \left( 
\begin{array}{cc}
4 & 1 \\ 
1 & 0%
\end{array}%
\right) \left( 
\begin{array}{cc}
3 & 1 \\ 
1 & 0%
\end{array}%
\right) \left( 
\begin{array}{cc}
1 & 1 \\ 
1 & 0%
\end{array}%
\right) \left( 
\begin{array}{cc}
3 & 1 \\ 
1 & 0%
\end{array}%
\right) \right.   \notag \\
&&\left. \left( 
\begin{array}{cc}
1 & 1 \\ 
1 & 0%
\end{array}%
\right) \left( 
\begin{array}{cc}
5 & 1 \\ 
1 & 0%
\end{array}%
\right) \right] \left( 
\begin{array}{c}
2 \\ 
1%
\end{array}%
\right) .  \label{b4k}
\end{eqnarray}
\end{expl}
\begin{expl} \label{ex6}
The relations between the continued fraction (\ref{b4d}) and the matrix (\ref%
{b4g}) with $r=4$, for instance, is 

\begin{eqnarray}
\left( 
\begin{array}{c}
k_{1} \\ 
k_{2}%
\end{array}%
\right)  &=&\left[ \left( 
\begin{array}{cc}
q_{2} & 1 \\ 
1 & 0%
\end{array}%
\right) \left( 
\begin{array}{cc}
q_{3} & 1 \\ 
1 & 0%
\end{array}%
\right) \left( 
\begin{array}{cc}
q_{4} & 1 \\ 
1 & 0%
\end{array}%
\right) \right] \left( 
\begin{array}{c}
k_{4} \\ 
k_{5}%
\end{array}%
\right)   \notag \\
&=&\mathbb{B}_{4}^{-1}\left( 
\begin{array}{c}
k_{4} \\ 
k_{5}%
\end{array}%
\right) =\left( 
\begin{array}{c}
q_{2}\left( k_{4}+k_{4}q_{3}q_{4}+k_{5}q_{3}\right) +k_{4}q_{4}+k_{5} \\ 
\left( k_{4}+k_{5}q_{3}+k_{4}q_{3}q_{4}\right) 
\end{array}%
\right) .  
\label{b4l}
\end{eqnarray}
and the ratio $k_{1}/k_{2} $ can be written as a continued fraction, 
\begin{eqnarray}
\frac{k_{1}}{k_{2}} &=&q_{2}+\frac{k_{4}q_{4}+k_{5}}{k_{4}+q_{3}\left(
k_{4}q_{4}+k_{5}\right) }=q_{2}+\cfrac{1}{q_{3}+\cfrac{k_{4}}{k_{4}q_{4}+k_{5}}%
}  \label{b4m} \\
&=&q_{2}+\cfrac{1}{q_{3}+\cfrac{1}{q_{4}+\cfrac{k_{5}}{k_{4}}}},  \notag
\end{eqnarray}
and in general one has Eq. (\ref{b4d}).
\end{expl}

The coefficients $p_l$ and $t_l$ are obtained by inverting Eqs. (\ref%
{b4h}) and (\ref{b4j}), 
\begin{equation}
\left( 
\begin{array}{cc}
q_{2} & 1 \\ 
1 & 0%
\end{array}%
\right) ^{-1}\left( 
\begin{array}{c}
k_{1} \\ 
k_{2}%
\end{array}%
\right) =\left( 
\begin{array}{c}
k_{2} \\ 
k_{1}-q_{2}k_{2}%
\end{array}%
\right) =\left( 
\begin{array}{c}
k_{2} \\ 
k_{3}%
\end{array}%
\right) ,  \label{b14}
\end{equation}
\begin{equation}
\left( \left( 
\begin{array}{cc}
q_{2} & 1 \\ 
1 & 0%
\end{array}%
\right) \left( 
\begin{array}{cc}
q_{3} & 1 \\ 
1 & 0%
\end{array}%
\right) \right) ^{-1}\left( 
\begin{array}{c}
k_{1} \\ 
k_{2}%
\end{array}%
\right) =\left( 
\begin{array}{c}
k_{1}-k_{2}q_{2} \\ 
k_{2}\left( q_{2}q_{3}+1\right) -k_{1}q_{3}%
\end{array}%
\right) =\left( 
\begin{array}{c}
k_{3} \\ 
k_{4}%
\end{array}%
\right)   \label{b15}
\end{equation}
and with the matrix 
\begin{equation}
\mathbb{B}_{4}:=\left[ \prod_{i=2}^{4}\left( 
\begin{array}{cc}
q_{i} & 1 \\ 
1 & 0%
\end{array}%
\right) \right] ^{-1}  \label{b15a}
\end{equation}
we get the relation, 
\begin{equation}
\mathbb{B}_{4} \left( 
\begin{array}{c}
k_{1} \\ 
k_{2}%
\end{array}%
\right) =\left( 
\begin{array}{c}
k_{2}\left( q_{2}q_{3}+1\right) -k_{1}q_{3} \\ 
k_{1}\left( q_{3}q_{4}+1\right) -k_{2}\left(
q_{2}+q_{4}+q_{2}q_{3}q_{4}\right) 
\end{array}%
\right) =\left( 
\begin{array}{c}
k_{4} \\ 
k_{5}%
\end{array}%
\right) .  \label{b16}
\end{equation}
The coefficients are 
\begin{eqnarray*}
\mathbb{B}_{3} = \left( 
\begin{array}{cc}
p_{1} & t_{1} \\ 
p_{2} & t_{2}%
\end{array}%
\right)  &=& \left( 
\begin{array}{cc}
1 & -q_{2} \\ 
-q_{3} & q_{2}q_{3}+1%
\end{array}%
\right)  \\
\mathbb{B}_{4} = \left( 
\begin{array}{cc}
p_{2} & t_{2} \\ 
p_{3} & t_{3}%
\end{array}%
\right)  &=& \left( 
\begin{array}{cc}
-q_{3} & q_{2}q_{3}+1 \\ 
q_{3}q_{4}+1 & -q_{2}-q_{4}-q_{2}q_{3}q_{4}%
\end{array}%
\right) ,
\end{eqnarray*}
thus the sequence of equations
\begin{eqnarray}
k_{1}-q_{2}k_{2} &=&k_{3}  \notag \\
-q_{3}k_{1}+\left( q_{2}q_{3}+1\right) k_{2} &=&k_{4}  \label{b18} \\
k_{1}\left( q_{3}q_{4}+1\right) -\left( q_{2}+q_{4}+q_{2}q_{3}q_{4}\right)
k_{2} &=&k_{5}  \notag
\end{eqnarray}
represents the relations $p_{l}k_{1}+t_{l}k_{2}=k_{l+2}$ with $%
l=1,..,,r$. Thus we can write 
\begin{equation}
\mathbb{B}_{r} \left( 
\begin{array}{c}
k_{1} \\ 
k_{2}%
\end{array}%
\right) =\left( 
\begin{array}{cc}
p_{r-1} & t_{r-1} \\ 
p_{r} & t_{r}%
\end{array}%
\right) \left( 
\begin{array}{c}
k_{1} \\ 
k_{2}%
\end{array}%
\right) =\left( 
\begin{array}{c}
k_{r+1} \\ 
k_{r+2}%
\end{array}%
\right)   \label{b19}
\end{equation}
and the coefficients $p_{l}$ and $t_{l}$, see Eqs. (\ref{b11}) and (\ref{b26}%
), depend only on the quotients, as shown in Table \ref{b20} 
and we can confront these expression with those calculated in Example \ref%
{ex3}, for instance, $\left( q_{2},...,q_{6}\right) =\left( 4,1,2,3,21\right) $.
\begin{table}[htbp]
\begin{tabular}{|l||c|c|}
\hline
$l$ & $p_{l}$ & $t_{l}$ \\ \hline\hline
$1$ & $1$ & $-q_{2}$ \\ \hline
$2$ & $-q_{3}$ & $q_{2}q_{3}+1$ \\ \hline
$3$ & $q_{3}q_{4}+1$ & $-\left( q_{2}\left( 1+q_{3}q_{4}\right)
+q_{4}\right) $ \\ \hline
$4$ & $-q_{3}\left( 1+q_{4}q_{5}\right) -q_{5}$ & $%
q_{2}q_{3}+q_{2}q_{5}+q_{4}q_{5}+q_{2}q_{3}q_{4}q_{5}+1$ \\ \hline
\end{tabular}%
\caption{\small{The coefficients in terms of the quotients.}}
\label{b20}
\end{table}
%
\section{Two strategies for a more secure Alice-Bob communication} 
\label{strategy}
%
\subsection{Choosing the exponents $k_1$ and $k_2$}\label{choosing}
%
For the construction of the matrices $\mathbb{M}_{k_{1}}$ and $\mathbb{M}%
_{k_{2}}$, we consider that the best strategy for Alice (or Bob) consists 
in choosing a certain sequence of positive integers $\left(
q_{2},q_{3},...q_{r},q_{r+1}\right) $ and then to calculate the continued 
fraction $\left/ \left/q_{2},q_{3},...q_{r},q_{r+1}\right/ \right/ $, from which she (or he) 
obtains the coprime numbers $k_{1}$ and $k_{2}$ to be used as exponents. We envisage this
procedure as a good strategy because according to a theorem exposed in \cite%
{knuth2}, in the sequence $\left( q_{2},q_{3},...,q_{8}\right) $ for the
quotients in Euclid algorithm, the approximate probability for one integer $%
q_{j}$ to take the value $a$ is 
\begin{equation}
p\left( a\right) =\log _{2}\left( 1+\frac{1}{a}\right) -\log _{2}\left( 1+%
\frac{1}{a+1}\right) =\log _{2}\left( 1+\frac{\left( a+1\right) ^{2}}{\left(
\left( a+1\right) ^{2}-1\right) }\right) ,  \label{b5p}
\end{equation}
thus $p\left( 1\right) =0.41505$, $p\left( 2\right) =0.16993$, $p\left(
3\right) =0.09311$, $p\left( 4\right) =0.05889$, etc. So Alice (or Bob)
should choose quotients that contradict that pattern in order to fool an
eavesdropper, turning the decryption a hard task, i.e., more computational
time consuming. The statistics of appearance of the numbers $1$ to $4$ in a
sequence of quotients $\left( q_{2},q_{3},...q_{r},q_{r+1}\right) $, that
comes from a fraction $k_{1}/k_{2}$, with $k_{1}$ and $k_{2}$ being coprime
numbers picked at random, has probability $74\%$. Therefore, if Alice (or
Bob) believes that to guess the matrix $\mathbb{M}$ the eavesdropper Eve is
going to use mostly the numbers $1$ to $4$, thus Alice (or Bob) have the option 
to construct $k_{1}$ and $k_{2}$ using most of the quotients $q_{j}$ higher 
than $4$. 
%
\subsection{Disguising the ratio $k_1/k_2$}\label{disguising}
%
The matrix (\ref{a3}) has determinant $19+4i$ and trace $-i$; raising it to the 
two prime numbers $k_{1}=19$ and $k_{2}=13$, for instance, results in 
\begin{equation}
\mathbb{C}=\left( 
\begin{array}{cc}
513\, 852\,777\, 463- i\, 2298\, 638\,980\, 668 & 
-2931\, 007\,733\,267-i\, 5289\, 191\,614\,510 \\ 
-4739\, 850\,067\,058-i\, 3755\, 020\,054\,445 & 
-14\,804\, 380\,518\,615-i\, 191\, 643\,568\,579%
\end{array}%
\right) ,  \label{aa3}
\end{equation}
and
\begin{equation}
\mathbb{D}=\left( 
\begin{array}{cc}
14\,127\,153+ i\, 166\,148\,948 & 378\,280\,035+i\, 217\,438\,622 \\ 
432\,811\,810+i\,55\,220\,253 & 881\,816\,207-i\, 535\,190\,869%
\end{array}%
\right) .  \label{aa4}
\end{equation}
Utilizing the Algorithm \ref{alg1} one retrieves the original matrix $\left( 
\mathbb{C}^{-1}\right) ^{2}\mathbb{D}^{3}=\mathbb{M}$. However, there is 
a weak point that makes the task easier to a cryptanalyst: he has a direct 
access to the ratio $k_{1}/k_{2}$ since $\left( \ln \left( \det 
\mathbb{C}\right) \right) /\left( \ln \left( \det \mathbb{D}\right) \right) 
=k_{1}/k_{2}$, and in the current numerical example the result is
\be 
\frac{\ln \left( \det \mathbb{C}\right)}{\ln \left( \det \mathbb{D}\right)}
\approx 1.4502-i\ 0.1622\, ,
\ee
which is a strong clue to guess the exponents since the real part is quite 
close to $19/13\approx 1.4615$. Another relation relies on the difference 
$\ln \left( \det \mathbb{C}\right) -\ln \left( \det \mathbb{D}\right) =
\left( k_{1}-k_{2}\right) \ln \left( \det \mathbb{M}\right) $; but the 
presence of the unkown $\det \mathbb{M}$ does not introduce an additional 
advantage to guess $\left( k_{1}-k_{2}\right) $. 

In order to blur the relation $\left( \ln \left( \det \mathbb{C}\right)
\right) /\left( \ln \left( \det \mathbb{D}\right) \right) =k_{1}/k_{2}$, 
Alice can adopt a simple additional procedure, without affecting sensibly the 
sequence of steps of the Algorithm \ref{alg1}, making it part of the protocol. 
Instead of placing publicly the matrices $\mathbb{C}$ and $\mathbb{D}$, 
she will place other two associated matrices as to be explained below, 
although without ruling out that other simple procedures can be invented and 
be more efficient. Firstly she chooses two positive integer numbers, $m$ and $l$, 
and calculates the inverse modulo $m$ equation for $k_{1}$ and $k_{2}$, i.e.
$\left( k_{i}x_{i}\right) \modu m=1$, $i=1,2$, so that she gets the class of 
equivalent solutions, $x_{i}\left( l\right) = ml+x_{i}\left( 0\right) $, 
$l=0,1,2,...\, $. Thus the inverse modulo of $k_{i}$ is 
$\left[ k_{i}^{-1}\right] =x_{i}\left( l\right) $ and any value $l$ can be picked 
by Alice (or Bob), $\bar{k}_{i}:=x_{i}\left( l\right) $, and she defines the 
complex number $z=\bar{k}_{1} +i\ \bar{k}_{2}$. A set of $L$ matrices is constructed 
and displayed publicly, $\mathcal{M} = \left\{ \mathbb{Z}\left( z, z^*\right) \right\}$. 
These matrices have the same dimension of matrix $\mathbb{M}$, they are different 
from each other, they are enumerable (as $p_1,...p_r,...p_s,...p_L$) and the argument 
$z$ is not defined numerically. Alice then picks two matrices of the set $\mathcal{M}$ 
that she calls $\mathbb{Z}_{C}\left( z, z^*\right) $ and $\mathbb{Z}_{D}\left( z, z^*\right) $, 
where they now depend on the numerical values $\bar{k}_{1},\bar{k}_{2}$. In sequence 
Alice calculates the matrices $\mathbb{\bar{C}=Z}_{C}\mathbb{C}$ and 
$\mathbb{\bar{D}=Z}_{D}\mathbb{D}$ that she places publicly. Now the key Alice 
has to share with Bob, in order to enable him to get $\mathbb{M}$, is constituted by 
the 6-tuple $\left( k_{1},k_{2},m,l,p_r,p_s \right):= K $, that he uses to calculate 
$\left( \left( \mathbb{Z}_{C}^{-1}\mathbb{\bar{C} }\right) ^{-1}\right) ^{2}
\left( \mathbb{Z}_{D}^{-1}\mathbb{\bar{D}}\right) ^{3}=\mathbb{M}$. If a 
cryptanalyst calculates the ratio $\left( \ln \left( \det \mathbb{\bar{C}}\right) 
\right) /\left( \ln \left(\det \mathbb{\bar{D}}\right) \right) $ she/he will not 
extract an obvious information about the ratio $k_{1}/k_{2}$, as before, due 
to the presence of $\ln \det \left( \mathbb{Z}_{C}\right) $ and $\ln \det 
\left( \mathbb{Z}_{D}\right)$, 
\begin{equation}
\frac{\ln \left( \det \mathbb{\bar{C}}\right) }{\ln \left( \det \mathbb{\bar{%
D}}\right) }=\frac{\ln \det \left( \mathbb{Z}_{C}\right) +k_{1}\ln \det 
\mathbb{M}}{\ln \det \left( \mathbb{Z}_{D}\right) +k_{2}\ln \det \mathbb{M}}.
\label{r1}
\end{equation}
\begin{expl} \label{ex6a}
Alice adopts $m=11$ and $l=2$: $\bar{k}_{1}:=x_{1}\left(
2\right) =11\times 2+7=29$ and  $\bar{k}_{2}:=x_{2}\left( 2\right) =11\times
2+6=28$; she defines $z=\bar{k}_{1}+i\bar{k}_{2}=29+i28$ and chooses, for instance, 
from the set $\mathcal{M}$ the matrices,
\begin{equation}
\mathbb{Z}_{C}\left( z, z^*\right) =\left( 
\begin{array}{cc}
z & 1 \\ 
-1 & z^{\ast }%
\end{array}%
\right) ,\quad \mathbb{Z}_{D}\left( z, z^*\right) =\left( 
\begin{array}{cc}
i & z \\ 
z^{\ast } & -i%
\end{array}%
\right) ,  \label{zz}
\end{equation}
then the ratio (\ref{r1}) becomes
\begin{equation}
\frac{\ln \left( \det \mathbb{\bar{C}}\right) }{\ln \left( \det \mathbb{\bar{%
D}}\right) }=\frac{\ln \left( \left\vert z\right\vert ^{2}+1\right)
+k_{1}\ln \det \mathbb{M}}{\ln \left( -\left\vert z\right\vert ^{2}+1\right)
+k_{2}\ln \det \mathbb{M}} \approx 1.0456 -i\, 5.696\times 10^{-2},
\label{r2}
\end{equation}
whose real part does not give a hint about the ratio $k_{1}/k_{2}\approx 1.4615$, 
as in the case where Alice does not uses in her protocol the matrices 
$\mathbb{Z}_{C}$ and $\mathbb{Z}_{D}$. This kind of disguising is not very 
costly: the multiplications of $\mathbb{C}$ and $\mathbb{D}$ by $\mathbb{Z}_{C}$ 
and $\mathbb{Z}_{D}$ add only two additional operations for Alice, while for Bob, 
the task to retrieve $\mathbb{M}$ needs the use of $\mathbb{Z}_{C}^{-1}$ and 
$\mathbb{Z}_{D}^{-1}$, that introduces four additional operations: two inversions 
and two multiplications. Worth to note that the 6-tuple key $K$ can also be split 
into two (or more) keys $K_1$ and $K_2$ that complement each other.
\end{expl} 
%
\section{Upper bound for the number operations to construct a matrix $\mathbb{%
M}^{k}$}\label{upperbound}
%
The number of operations, $\mathcal{N}_{op}\left( k\right) $, Alice has to
perform (the product of $k$ matrices, namely $\mathbb{M}^{k}$) is $k-1$ \emph{%
without storage} of partial products (in the example \ref{ex2}, there should
be $1000$ of such products). Now, if Alice stores the already multiplied
matrices, namely $\mathbb{M}^{2}$, $\mathbb{M}^{3}...\mathbb{M}^{k-1}$ and
uses them for further operations, the necessary number of operations to
calculate $\mathbb{M}^{k}$ is significantly reduced, we estimate that $ 2%
\sqrt{k}$ is an upper bound for $k\gg 1$, as we presume that other approaches could 
reduce this number. To prove this result we assume that for a given $%
\mathbb{M}$ one calculates $\mathbb{M}^{2}$ and store it, so the available 
matrices are $S_{2}=\left\{ \mathbb{M},\mathbb{M}^{2}\right\} $ and the
number of multiplications is reduced from $k-1$ to 
\begin{equation}
\mathcal{N}_{op}\left( k|S_{2}\right) =1+\left[ \frac{k-1}{2}\right] \ ,
\label{c1}
\end{equation}
for $k>2$, where the square brackets $\left[ x\right] $ means the integer 
part of $x$. Thus, for $k=113$, for instance, the number of multiplications 
is reduced from $112$ without storage to 57 with $S_{2}$ storage, as shown 
in Table \ref{c2}. For 
$S_{3}=\left\{ \mathbb{M},\mathbb{M}^{2},\mathbb{M}^{3}\right\} $ 
we have to perform two operations to produce $\mathbb{M}^{2}$ and 
$\mathbb{M}^{3}$ to be then stored. Thus, for $k>3$ the number of operations 
is reduced from $57$ (storage using the set $S_{2}$) to $39$ when one uses 
the set $S_{3}$, as shown  in Table \ref{c2}.
\begin{equation}
\mathcal{N}_{op}\left( k|S_{3}\right) =2+\left[ \frac{k-1}{3}\right]\ .
\label{c3}
\end{equation}
\begin{table}[H]
\centering
\begin{tabular}{|c|c|c|c|c|c|c|c|c|c|c|l|l|l|}
\hline
$k$ & $3$ & $4$ & $5$ & $6$ & $7$ & $8$ & $9$ & $10$ & $11$ & $12$ & $\cdots 
$ & $113$ & $\cdots $ \\ \hline
$\mathcal{N}_{op}\left( k|S_{2}\right) $ & $2$ & $2$ & $3$ & $3$ & $4$ & $4$
& $5$ & $5$ & $6$ & $6$ & \multicolumn{1}{c|}{$\cdots $} & 
\multicolumn{1}{c|}{$57$} & \multicolumn{1}{c|}{$\cdots $} \\ \hline
$\mathcal{N}_{op}\left( k|S_{3}\right) $ & $ $ & $3$ & $3$ & $3$ & $4$ & $4$
& $4$ & $5$ & $5$ & $5$ & \multicolumn{1}{c|}{$\cdots $} & 
\multicolumn{1}{c|}{$39$} & \multicolumn{1}{c|}{$\cdots $} \\
\hline
\end{tabular}%
\caption{\small{Number of operations for calculating $\mathbb{M}^{k}$ 
with $S_2$ and $S_3$ storages.}}
\label{c2}
\end{table}

In general, for any number of stored matrix products $x_{s}$ the number 
of multiplications to construct $\mathbb{M}^{k}$ is 
\begin{equation}
\mathcal{N}_{op}\left( k|S_{x_{s}}\right) =\left( x_{s}-1\right) +\left[ 
\frac{k-1}{x_{s}}\right] ,  \label{c5}
\end{equation}
where the term in parenthesis stands for the number of multiplication
necessary to construct the set $S_{x_{s}}$. For each $k$ we can find 
the optimum number of products to be stored in order to minimize the 
number of matrix multiplications by looking for the point of minimum 
of Eq. (\ref{c5}) 
\begin{equation}
\frac{d\mathcal{N}_{op}\left( k|S_{x_{s}}\right) }{dx_{s}}=1-\left[ \frac{k-1%
}{x_{s}^{2}}\right] =0\Longrightarrow \left[ x_{s,\min }\right]:=\bar{k} =\left[ 
\sqrt{k-1}\right]   \label{c6}
\end{equation}
and 
\begin{equation}
\mathcal{N}_{op}\left( \bar{k}\right) :=\mathcal{N}_{op}\left( k|S_{\bar{k}} \right) 
= \left[ 2\sqrt{k-1}-1\right] \ ,  \label{c7}
\end{equation}
where $\bar{k} $ is the minimum number of powers of $\mathbb{M}$
necessary to construct the set $S_{\bar{k} }=\left\{ 
\mathbb{M},\mathbb{M}^{2},...,\mathbb{M}^{\bar{k}}\right\} $. 
This is exemplified in Table \ref{c8}. 
\begin{table}[htbp]
\begin{tabular}{|c|l|l|l|l|l|l|l|l|l|l|l|l|l|}
\hline
$k$ & $2$ & $3$ & $4$ & $5$ & $6$ & $7$ & $8$ & $9$ & $10$ & $11$ & $12$ & $%
\cdots $ & $113$ \\ \hline  &  &  &  &  &  &  &  & 
&  &  &  & $  $ & \\
$\bar{k} $ & $1$ & $1$ & $2$ & $2$ & $2$ & $2$ & $3$ & $3$
& $3$ & $3$ & $3$ & $ \cdots $ & $11$ \\ \hline
$\mathcal{N}_{op}\left( \bar{k}\right) $ & $1$ & $2$ & $2$ & $3$ & $3$ & $4$
& $4$ & $5$ & $5$ & $5$ & $6$ & $ \cdots $ & $20$ \\ \hline
\end{tabular}%
\caption{\small{Minimizing the number of operations for calculating $\mathbb{M}^{k}$ 
with storage $S_{\bar{k}}$ with $\bar{k} = \left[  \sqrt{k-1}\right]$.}}
\label{c8}
\end{table}
\begin{figure}[htbp]
\centering
\includegraphics[height=2.3in, width=3.5in]{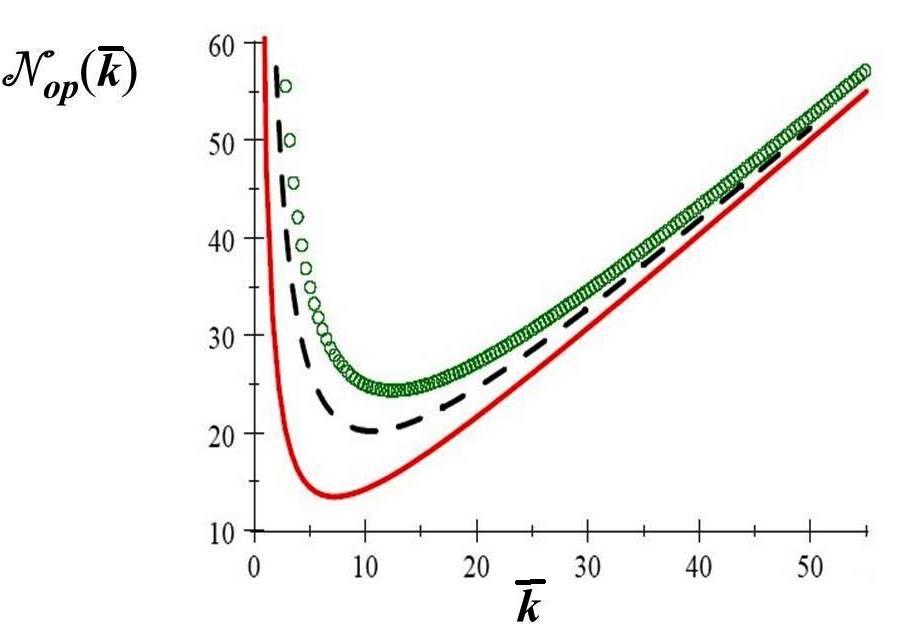}
\caption{\small {Eq. (\ref{c5}) for $k=53,113,159$ (red
solid line, black dashed and circles in green)}}
\label{fig01}
\end{figure}
For $k\gg 1$ we have $\bar{k} \approx \sqrt{k}$ and $%
\mathcal{N}_{op }\left( k\right) \approx 2\sqrt{k}$. In Fig. \ref{fig01}
we plot the graphs for $k=53,113,159$ and we find $\bar{k}
=7,10,12$, that imply the following lowest numbers of necessary operations 
$\mathcal{N}_{op}\left( \bar{k}\right) \approx 13,20,24$.

In Fig. \ref{fig02} we plot the number of $k$ operations Alice has to
perform without (dashed line) and with (solid line) storage. Indeed, the
result, that we consider being only an upperbound, is quite remarkable in
terms of reducing the number of operations. Nonetheless, we did not considered here the
necessary computational time for storing and calling back a matrix. 
\begin{figure}[tbp]
\centering
\includegraphics[height=2.2in, width=3.3in]{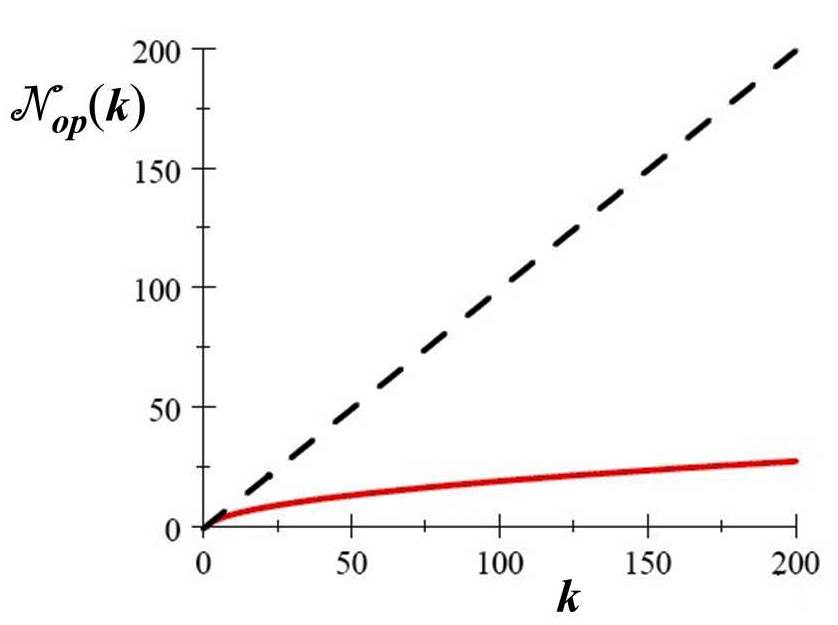}
\caption{{\small {The solid line stands for the number of operations
with storage ($\left[ 2\sqrt{k-1}-1\right]$) versus the number of operations
without storage $\left( k-1\right) $, dashed line.}}}
\label{fig02}
\end{figure}
%
\section{Some strategies of attack that a cryptanalyst may choose to 
recover $\mathbb{M}$}
%
A pertinent question that can be asked is: how a cryptanalyst could get the matrix 
$\mathbb{M}$ by only knowing the matrices $\mathbb{C}$ and $\mathbb{D}$? To answer this 
question, some considerations are advanced, as discussed in the Algorithm \ref{alg1}. 
Below we do not consider the extension of the protocol as proposed in subsection 
\ref{disguising}:
\begin{enumerate}[label=(\alph*)]
\item Eve has the matrices $\mathbb{C}$ and $\mathbb{D}$ of Example \ref{ex1},
for instance, but she is in the dark about their meaning, or whether they
contain any message. In this case she has no clue to follow in order to identify
some pattern, the only obvious property is that they commute. She shall try
to find some kind of correlation between the entries of the matrices,
between its inverses, or between the eigenvalues and eigenvectors. Otherwise
she waits for another two matrices, $\mathbb{C}^{\prime }$ and $\mathbb{D}%
^{\prime }$, to be conveyed by Alice and then tries to find some pattern.

\item Eve have access to the matrices $\mathbb{C}$ and $\mathbb{D}$, she knows
that they originate from a single matrix $\mathbb{M}$, of dimension $N$, 
raised to the powers $k_{1}$ and $k_{2}$ that may be the carrier of some message, 
but $k_{1}$ and $k_{2}$ are unknown to her. She can embrace the strategy of 
the direct approach by diagonalizing the matrices, such to get access to the 
eigenvalues and eigenvectors. The eigenvalues $\left\{ \lambda _{i}^{C}\right\} $ 
and $ \left\{ \lambda _{i}^{D}\right\} $ ($i=1,...,N$) are related to the 
eigenvalues of matrix $\mathbb{M}$, $\left\{ \lambda _{i}^{M}\right\} $, namely 
$\lambda _{i}^{C}=\left( \lambda _{i}^{M}\right) ^{k_{1}}$ and $\lambda
_{i}^{D}=\left( \lambda _{i}^{M}\right) ^{k_{2}}$, then Eve should try 
many integer values $p$ and $t$ to solve $\left( \lambda
_{i}^{C}\right) ^{1/p}$ and $\left( \lambda _{^{i}}^{D}\right) ^{1/t}
$ ($p,t=K,K+1,...,K+L$, where $K$ and $L$ are integer numbers chosen
arbitrarily by Eve) in order to find a common eigenvalue. However, Eve have
to deal with the multiplicity of roots, $\mu _{i,r}^{ C}$ and 
$\mu _{j,s}^{ D }$ with $i,j$ standing for the eigenvalues and $r=1,...,p$ 
and $s=1,...,t$ are the indices that classify the multiplicity of the roots 
for each eigenvalue. The number of eigenvalues to be analyzed is 
$N^2\times \left( L-K+1\right)^2$, which is time and, as well as, memory 
consuming, compared to Bob's task to get $\mathbb{M}$ once he receives 
the numbers $k_{1}$ and $k_{2}$.

\item Eve knows that each matrix, $\mathbb{C}$ and $\mathbb{D}$, originates
from the exponents $k_{1}$ and $k_{2}$ of a single matrix $\mathbb{M}$ that is
the carrier of a message, but $k_{1}$ and $k_{2}$ are unknown to her. Then
another strategy consists in solving the equations $\mathbb{C}-\mathbb{X}%
_{C}^{l_{C}}=0$ and $\mathbb{D}-\mathbb{X}_{D}^{l_{D}}=0$, where many
positive integers $l_{C}$ and $l_{D}$ should be tried for the unkown matrices
$\mathbb{X}_{C}$ and $\mathbb{X}_{D}$. If the matrices $%
\mathbb{C}$ and $\mathbb{D}$ belong to $C^{N\times N}$, then $2N^{2}$ (the
factor $2$ is present because the entries belong to the field of complex
numbers) polynomial algebraic equations, having degrees $l_{C}$ and $l_{D}$ 
respectively, must be solved. Each one of them has the same
number, $2N^{2}$, of unknown variables. Each matrix equation
accepts $L_{i}$ ($i=C,D$) different sets of solutions $\left\{ \mathbb{\bar{X}}%
_{C,r}\right\} $ and $\left\{ \mathbb{\bar{X}}_{D,r}\right\} $, $%
r=1,...,L_{i}$ if $\mathbb{C}$ and $\mathbb{D}$ do not show any kind of
symmetry or constraints that could reduce the number of independent
equations. For each pair $\left( l_{C},l_{D}\right) $ Eve has to verify
whether one solution, of the $2N^{2}$ possible ones, of each set of equations is
common to both; if that happens then Eve has successfully 
uncovered the matrix $\mathbb{\bar{M}}$, once she concluded that 
$\left( \bar{l}_{C},\bar{l}_{D}\right) =\left( k_{1},k_{2}\right) $, such 
that $\mathbb{C=M}^{k_{1}}$ and $\mathbb{D=M}^{k_{2}}$. Worth to mention
that Bob can use that very same method to get the right matrix because he
knows beforehand the values of $k_{1}$ and $k_{2}$, and he does not need, as Eve 
needs, to browse through many pairs of numbers to be used as exponents.
However, this approach is more time consuming than the proposition of our
algorithm needs and depends on high computational capacity. 

\item Eve knows that $\mathbb{C}=\mathbb{M}^{k_{1}}$ and $\mathbb{D}=\mathbb{M}%
^{k_{2}}$ and she is acquainted with the algorithm, but she ignores the values 
of the exponents $k_{1}$ and $k_{2}$ $\left( k_{1}>k_{2}\right)$. Knowing the
algorithm she decides to calculate the product of matrices 
$\mathbb{C}\left( \mathbb{D}^{-1}\right) ^{m_{1}}\equiv \mathbb{M}_{k_{3}}
\left( m_{1}\right) $, and as she doesn't know which is the right exponent she 
tries the values $m_{1}=1,...,L$. She repeats the process described in Eq. 
(\ref{B5}) sequentially for each value $m_{i}$, as it is explicitly presented 
in Table \ref{t1}.  
\begin{table}[htbp]
\begin{tabular}{|c|c|c|}
\hline
{\small Computing a sequence of matrices} & \# of {\small matrices} & 
{\small \# of matrix} {\small mult.} \\ \hline \hline
\multicolumn{1}{|l|}{$\mathbb{M}_{k_{3}}\left( m_{1}\right) \equiv \mathbb{C}%
\left( \mathbb{D}^{-1}\right) ^{m_{1}}$} & $L$ & \multicolumn{1}{l|}{$%
\sum_{m_{1}=1}^{L}m_{1}$} \\ \hline
\multicolumn{1}{|l|}{$\mathbb{M}_{k_{4}}\left( m_{1},m_{2}\right) \equiv 
\mathbb{D}\left( \left( \mathbb{M}_{k_{3}}\left( m_{1}\right) \right)
^{-1}\right) ^{m_{2}}$} & $L^{2}$ & \multicolumn{1}{l|}{$\left(
\sum_{m_{2}=1}^{L}m_{1}\right) ^{2}$} \\ \hline
\multicolumn{1}{|l|}{$\mathbb{M}_{k_{5}}\left( m_{1},m_{2},m_{3}\right)
\equiv \mathbb{M}_{k_{3}}\left( m_{1}\right) \left( \left( \mathbb{M}%
_{k_{4}}\left( m_{1},m_{2}\right) \right) ^{-1}\right) ^{m_{3}}$} & $L^{3}$
& \multicolumn{1}{l|}{$\left( \sum_{m_{1}=1}^{L}m_{1}\right) ^{3}$} \\ 
\hline
$\cdots $ & $\cdots $ & \multicolumn{1}{l|}{$\cdots $} \\ \hline
\multicolumn{1}{|l|}{$\mathbb{M}_{k_{r+1}}\left( m_{1},...,m_{r-1}\right)
\equiv $} &  & \multicolumn{1}{l|}{} \\ 
$\mathbb{M}_{k_{r-1}}\left( m_{1},...,m_{r-3}\right) \left( \left( \mathbb{M}%
_{k_{r}}\left( m_{1},...,m_{r-2}\right) \right) ^{-1}\right) ^{m_{r-1}}$ & $%
L^{r-1}$ & $\left( \sum_{m_{1}=1}^{L}m_{1}\right) ^{r-1}$ \\ \hline
\end{tabular}
\caption{\small{Sequence of matrices to be computed, the second column shows the number of 
matrices and the third column stands for the number of matrix multiplications.}}
\label{t1}
\end{table}
She then computes and analyze the matrices, stopping at some value $r>2$ when she 
succeeds in finding the seed matrix in the last group of $L^{r-1}$, after 
having checked $\sum_{l=1}^{r-1}L^{l}=\left( L\left( L^{r-1}-1\right) \right) /\left(
L-1\right) $ matrices. The analysis: Eve considers the matrices $\mathbb{M}_{k_{l}}\left(
m_{1},...,m_{l-2}\right) $, with $l=3,...,r+1$, and raises each one to
powers that begin at some positive integer $K_{1}$ and stops at a value $%
K_{2}$, $k=K_{1},K_{1}+1,..,K_{2}$, chosen arbitrarily ($\left( \mathbb{M}%
_{k_{l}}\left( m_{1},...,m_{l-2}\right) \right) ^{k}$). She then begins to quest 
whether an obtained matrix coincides with matrix $\mathbb{C}$ or
with $\mathbb{D}$. In the case she succeeds, she then acknowledges that she 
got the seed matrix $\mathbb{M}$. However, this procedure is of exponential 
complexity since for $L=9$ and $r=10$ she has to analyze $435,848,049\approx 
4.4\times 10^{8}$ matrices. For the same $L$ and for $r=15$ the number of 
matrices goes to $25,736,391,511,830\approx 2.6\times 10^{13}$; this number 
grows exponentially with $r$. Besides, if she is sure that she did not find the 
meaningful matrix, she can try another possibility, namely, she changes 
$\mathbb{C\rightarrow D}^{-1}$ and $\mathbb{D}^{-1}\rightarrow 
\mathbb{C}$ and repeats the same procedure. 

We can now analyze this procedure in terms of probabilities as a function of
time. What is the probability that Eve achieves the identification of the right 
matrix $\mathbb{M}$ at the $s$-th attempt? i.e., the first $s-1$ computed matrices 
are the wrong ones. If she assumes that $\mathbb{M}$ is one out of the $%
\mathcal{N}=\left( L\left( L^{r-1}-1\right) \right) /\left( L-1\right) $
matrices, and all being equally likely, then she attributes the success 
probability $p=\mathcal{N}^{-1}$. The probability to find $\mathbb{M}$ 
at the $s$-th trial is $q^{s-1}p$ with $q=\left( 1-p\right) <1$, for 
the probability of not being the right matrix. Normalizing, the probability 
is given as $P_{s}=q^{s-1}p/\left( 1-q^{\mathcal{N}}\right )$, such that 
$\sum_{s=1}^{\mathcal{N}}P_{s}=1 $. The estimated average time for Eve to 
discover the right matrix then is 
\begin{equation}
T\left(\mathcal{N} \right) =\tau _{0}\frac{\sum_{s=1}^{\mathcal{N}}sP_{s}}{%
\sum_{s=1}^{\mathcal{N}}P_{s}}\approx \tau _{0}\mathcal{N}\ ,
\label{meantime}
\end{equation}
a time that grows exponentially with $r$, where $\tau_{0}$ is an arbitrary 
unit of time. We do not rule out the possibility of creation of 
ingenuous strategies, more efficient than those we proposed here, to extract 
the seed matrix $\mathbb{M}$ from matrices $\mathbb{C}$ and $\mathbb{D}$, in a 
comparatively quite shorter time. 
\end{enumerate}
If Alice extends her protocol by blurring the matrices $\mathbb{C}$ and 
$\mathbb{D}$ as proposed in subsection \ref{disguising}, we have not been 
able to devise a numerical procedure for a cryptanalyst to extract matrices 
$\mathbb{Z}$ from matrices $\bar{\mathbb{C}}$ and $\bar{\mathbb{D}}$.
%
\section{Public key distribution}
%
Here we present two methods for the public key distribution, whose discussion is 
included because the encryption and then decryption are essential for the 
distribution of the 6-tuple key $K$, as discussed in subsection \ref{disguising}.
%
\subsection{Diffie-Hellman proposal}
%
In Ref. \cite{diffie}, W. Diffie and M. E. Hellman (DF) reported a method by which 
two, or even $n$, persons could exchange information through a secure connection 
created over a public channel and using a publicly divulged method only. In fact, 
it seems that the idea was already proposed and used by the British intelligence 
agencies, but it was not publicized \cite{singh}. The proposal was later improved and 
turned operational by R. L. Rivest, A. Shamir and L. Adleman, as reported in Ref. 
\cite{RSA}, after getting a patent for the method. As recognized by these authors, 
\aspe Their [DH] article motivated our research, since they presented the concept 
but not any practical implementation of such a system\aspd. Indeed, although the 
RSA method was a new encryption approach, the principle of the public key 
distribution was first publicized by Diffie and Hellman for civilian use purposes, 
although it was recently reported that it fails in practice \cite{adrian}. With 
this in mind we present below the principle of a public key distribution, each 
part (person) conserving a private or secret key, and how a confidential plaintext 
can be shared by the parts. This procedure can be used in our proposal of encrypting 
a matrix $\mathbb{M}$, containing a plaintext, by raising it to two powers, 
$k_1$ and $k_2$, and complementing it with the protocol exposed in subsection 
\ref{disguising}. The calculated matrices are the recipient of an  
encrypted publicly exposed message, and the 6-tuple key $K$ is essential for the 
encryption and decryption, so to get the seed matrix $\mathbb{M}$. As Alice has a 
message to be conveyed to Bob, she looks for a secure channel to send him the key $K$
and, for instance, they decide to use the DH method: 

The positive integer $K$ can be defined as $K:=l\modu p$, where $l$ and $p$ are also 
positive integers. The number $l$ is expressed as $q^{ab}$ where $q$, $a$ and $b$ 
are positive integers too, thus,
\be
K=q^{ab}\modu p=\left( q^{a}\right) ^{b}\modu p=\left( q^{b}\right)
^{a}\modu p\ . 
\ee
Alice and Bob adopt the pair of numbers ($p,q$) as a public key and Bob creates 
his own secret key $b$ and composes a number $k_{b}=q^{b}\modu p$ that he exposes 
publicly, although aimed to Alice. As she receives $k_{b}$ she chooses a number 
$a$ and calculates $\left( k_{b} \right)^{a} \modu p = K$. If she judges that 
this number $K$ is convenient she keeps $a$ as her secret key, otherwise she 
goes after another one. Once Alice chose a number $a$ this means that she decided 
what should be the common key $K$, then she calculates $k_{a}=q^{a}\modu p$ 
and sends it to Bob, without being concerned about its secrecy. 
Bob reads it and calculates $\left( k_{a}\right) ^{b}\modu p = K$, 
which is the symmetric operation done by Alice. Now Alice and Bob share 
a same number $K$. 
\begin{expl}\label{ex7} Let us admit that Alice and Bob want to share more than 
one key, by splitting $K$ into $K_1$ and $K_2$, that complement each other. 
The knowledge of $K_1$ with no knowledge of $K_2$ (and vice versa) is sterile. 
\newline
Alice (or Bob) produces a public key, for instance, 
$(p,q) = \left( 457,7832\right)$, then Bob chooses his secret key, for 
instance $b_{1}=2$, and calculates the number $k_{b_{1}}:=(7832)^{2}\modu 457 
= 313$ that he sends to Alice. As she reads $k_{b_{1}}$ she has to choose $%
a_{1}$  judiciously (in line with the strategies presented in section 
\ref{strategy}, for instance) such that $K_{1}$ becomes a common key to be 
shared with Bob. She considers $k_{b_{1}}$ and $p=457$ (both being prime numbers)
and calculates $K_{1}=\left( k_{b_{1}}\right) ^{a_{1}}\modu p$ and tries several 
values of $a_{1}=x$ in $K_{x}=\left( \left( 313^{x}\right) \modu 457\right)$, as
candidates for her secret key, and then decides to pick $a_{1}=11$ that gives 
$K_{1}=361$. She now calculates $k_{a_{1}}=q^{a_{1}}\modu p$ 
($k_{a_{1}}=\left( 7832\right) ^{11}\modu 457=438$) and keeps it. Alice sends 
then to Bob the second part of the key, $K_{2}$, and she chooses $a_{1}=7$; thus 
$ K_{2}=\left( \left( 313^{7}\right) \modu 457\right) =79$. She calculates 
the matrices $\bar{\mathbb{C}}$ and $\bar{\mathbb{D}}$ that she divulges 
publicly, but aimed primarily to Bob. In sequence she divulges publicly the numbers 
$k_{a_{1}}$ and $k_{a_{2}}=\left( 7832\right) ^{7}\modu 457=203$. 
As Bob reads them he computes $K_{1} = \left( k_{a_{1}}\right) ^{b_{1}}\modu p =
\left( 438\right) ^{2}\modu 457 = 361$ and $K_{2} =\left( 203\right) ^{2}\modu 457=79$, 
and he eventually shares with Alice the keys $K_{1}$ and $K_{2}$. He then effectuates 
the reversal operation discussed in subsection \ref{disguising} and then goes on with 
the Algorithm\ref{alg1} to retrieve the matrix $\mathbb{M}$ from the matrices 
$\mathbb{C}$ and $\mathbb{D}$. See the schematization of the key distributions 
in Fig \ref{keydist}.
\end{expl}
\begin{figure}[htbp]
\centering
\includegraphics[height=3.5 in, width=5.0 in]{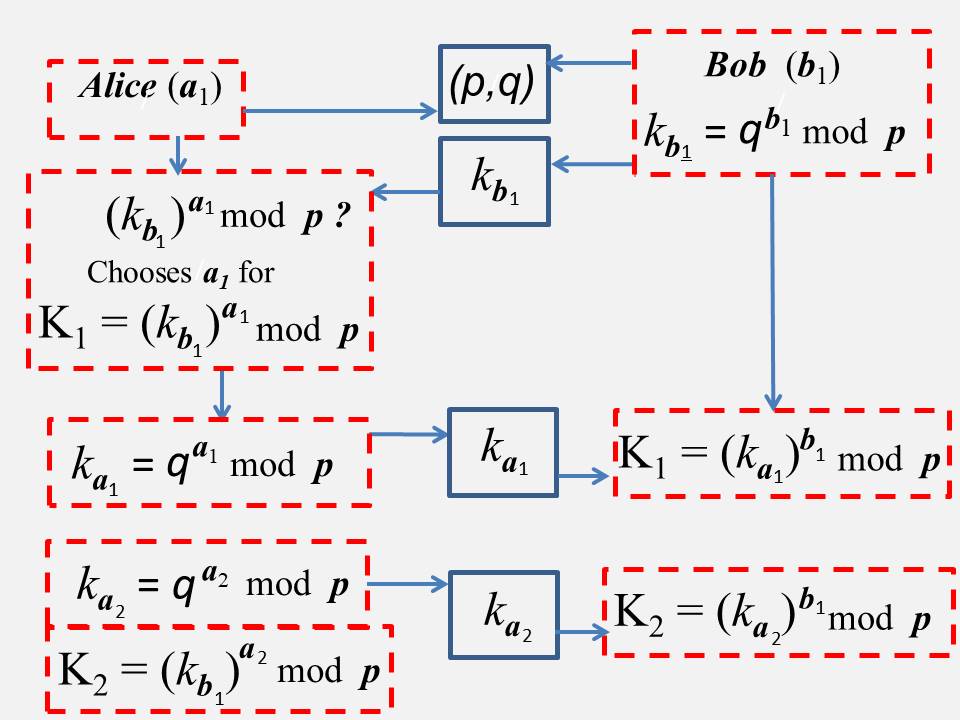}
\caption{{\small{Schematic picture of the key distribution between 
Alice and Bob. The solid line boxes represent the public codes and 
dashed boxes are for the personal calculations and the secret keys.}}}
\label{keydist}
\end{figure}
%
\subsubsection{Generalizing two by two secret keys for $N$ persons}
%
The Diffie-Hellman method is not restricted only to two persons (Alice and Bob), 
it can be generalized to involve $N$ persons, sharing, two by two, 
bilateral secret keys, $K_{i,j}$, $i,j=1,2,...N$, with symmetry $%
K_{i,j}=K_{j,i}$ and $K_{i,i}=0$, that makes $N\left( N-1\right) /2$
communication channels. The same public key ($p,q$) can be used by every person $P_i$
although he could choose different secret keys to be used for each recipient
of the message. Formalizing, we have
\be
K_{i,j}=q^{a_{i}a_{j}}\modu p=\left( q^{a_{i}}\right) ^{a_{j}}\modu%
p=\left( q^{a_{j}}\right) ^{a_{i}}\modu p\ . 
\ee
P$_{i}$ and P$_{j}$ share the public key $(p,q)$ used for producing their 
respective secret keys $a_{i}$ and $a_{j}$. P$_{i}$ composes the number 
$k_{a_{i}}=q^{a_{i}}\modu p$ that he sends publicly to P$_{j}$ that uses its 
own private key to calculate  
$\left( k_{a_{i}}\right) ^{b_{j}}\modu p=K_{i\longrightarrow j}:=K_{i,j}\ $. 
By its turn, using his secret key, P$_{j}$ composes a number $k_{a_{j}}=q^{a_{j}}\modu p$ 
and sends it to P$_{i}$ that does the symmetric operation by calculating 
$\left( k_{a_{j}}\right) ^{a_{i}}\modu p=K_{j\longrightarrow i}:=K_{j,i}$. 
As $\left( k_{a_{i}}\right) ^{a_{j}}\modu p=\left( k_{a_{j}}\right)
^{a_{i}}\modu p$, therefore P$_{i}$ and P$_{j}$ share the same number $%
K_{j,i}=K_{i,j}$.
%
\subsection{RSA method and Algorithm \ref{alg1}}
%
The method of encrypting a message invented by Rivest, Shamir and Adleman \cite{RSA}
has its security based, in part, on the difficulty of factoring large numbers, 
dispensing the use of a courier to carry secret keys. We present the RSA method 
aimed to transmit, from Alice to Bob, the numbers $K_1$ and $K_2$, that she used 
as exponents to rise the matrix $\mathbb{M}$ that contains a message. Bob is the character 
responsible for the construction of the keys to be used for the encryption 
and the decryption of $K_1$ and $K_2$.  

Bob chooses two ``many digits'' prime numbers, $p$ and $q$, and computes their 
product $n=pq$ ($n$ is called biprime number), he then considers the 
\emph{Euler totient} function $\varphi (n)$ (that gives the number of positive 
integers relatively prime to but less than $n$); in this case 
$\varphi (n)= \varphi (p)\varphi (q) = (p-1)(q-1)$, see \cite{cohn} for instance. 

Then Bob chooses an integer number $d$ that is relatively prime to $\varphi (n)$, 
$\gcd\left(d,\varphi (n)\right) = 1$, and also to $n$, $\gcd\left(d,n)\right) = 1$. 
Thus $d$ has an inverse modulo $\varphi (n)$, $[d^{-1}]:= e$, such that 
$ e\cdot d \equiv 1 (\modu \varphi (n))$, or $e\equiv d^{-1} (\modu \varphi (n))$, 
and $d \equiv e^{-1} (\modu \varphi (n))$. The r\^{o}le of $e$ within the 
method is to encrypt the message $M$ (shuffling the digits). Thereafter Bob makes 
the pair of numbers $(e,n)$ public. 

After reading $(e,n)$ Alice encrypts her message $M$ as $C =  M^{e}\modu n$ 
and makes it public, then Bob uses his secret decryption key 
$C^{d} \modu n = M$ to retrieve $M$ from $C$. The RSA method relies essentially 
on the difficulty that a cryptanalyst will have to decompose a large -- some 200 
digits -- biprime number $n$ into its two factors, as long as he does not know 
$d$. Below we illustrate the RSA method to be employed in conjunction with the 
reversal operation discussed in subsection \ref{disguising} 
and the Algorithm \ref{alg1}.
\begin{expl}\label{ex8}
Bob chooses the numbers $p=37$ and $q=53$, such that $n=p\cdot q = 1961$, 
$\varphi (n) =(q-1)(p-1)=1872 $ and turns $n$ public. He then picks $d=163$ 
as his secret decryption key, checks that $\gcd\left( \varphi (n),d \right) = 1$ 
is satisfied and calculates the encryption key $e =d\modu \varphi (n)=379$ that 
he turns public. Alice wants to send to Bob $K_1 = 361$ and $K_2 = 079$ 
without worrying whether they will be intercepted by an eavesdropper. 
Once she has free access to $e$ and $n$ she encrypts them as 
$ Y_1 = ((K_1)^{e}\modu n) = 324$ and $ Y_2 = ((K_2)^{e}\modu n) 
= 1253$, and makes these numbers public together with the matrices 
$\bar{\mathbb{C}}$ and $\bar{\mathbb{D}}$.

In order to get a swift access to the matrix $\mathbb{M}$ using the reversal 
operation in subsection \ref{disguising} and the Algorithm \ref{alg1}, Bob uses 
his secret key $d$ to get $K_1 = (Y_1)^{d}\modu n = 361$ and $K_2 = (Y_2)^{d}\modu 
n = 079$. In this way Alice was able to send to Bob the necessary keys  
that he uses to get access to matrix $\mathbb{M}$. It is worth noting that, 
otherwise, Alice could have sent a single number $361079$ instead of two if 
there was an understanding that the first half of the 
digits (beginning from the left) corresponds to $K_1=361$ and the other half 
to $K_2=079$.
\end{expl}
%
\section{Summary and conclusions}
%
We have here proposed a method to encrypt a matrix $\mathbb{M}$ and to decrypt it 
that is based on the calculation consisting in raising $\mathbb{M}$ to powers 
whose exponents are positive coprime numbers $k_1$ and $k_2$, producing the 
matrices $\mathbb{C}$ and $\mathbb{D}$, respectively. The knowledge of these 
numbers, that can be be considered as a key, makes the decryption 
(extracting $\mathbb{M}$ from $\mathbb{C}$ and $\mathbb{D}$) quite immediate 
in terms of a number of a certain number of operations, according to the Algorithm \ref{alg1}. 
In the case the key is unknown, we estimated that a third party could spend 
a much larger time to get $\mathbb{M}$ by the method of trial and error. 
Our algorithm follows the standard approach to cryptography: easy to encrypt 
and also to decrypt to those that have the key, but quite hard to a third 
party that has to guess it. The keys for the decryption, $k_1$ and $k_2$ 
can be shared between the parts by using either the DH method or the RSA, 
as described above.

Our proposal for encrypting a message is one possible application of the
discussed algorithm, nevertheless we believe that it can be useful in 
other instances. For example, in retrodicting a seed matrix $\mathbb{M}$ 
in a discrete $n-$step Markov process. If one only knows two matrices 
$\mathbb{M}_{k_{1}}$ and $\mathbb{M}_{k_{2}}$ (entries belonging to 
$\mathbb{R}$ in the closed interval $\left[0,1\right] $) describing the 
evolution at two times (known), $k_{1}$ and $k_{2} $ ($k_{1}>k_{2}$), 
once $\mathbb{M}_{1}$ is retrieved we are able to calculate the powers 
$\mathbb{M}^{k}$ for any integer $k\in \left[ 1,\infty\right) $, being 
thus able to run all the history of the changes of the matrix, so 
establishing a kind of determinism, although the entries are
conditional probabilities. 
%
\begin{appendix}
%
\section{Lemma and Theorem \label{theor}}
%
\begin{lemma}
\label{lema1}%
Euclid's algorithm for coprime numbers: given two positive integers, $k_{1}$
and $k_{2}$, with $k_{1}>k_{2}$, being coprime, ($\gcd \left( k_{1},k_{2}\right)
=1$), in modular arithmetic we have $k_{1}\modu\ k_{2}= k_{3}$, or
\begin{equation}
k_{1}=q_{1}k_{2}+k_{3},   \label{B1}
\end{equation}
$q_{1}$ being the quotient and $k_{3} \left(< k_{2} \right)$ is the remainder, 
then it follows that the numbers $k_{2}$ and $k_{3}$ are also coprime.
\end{lemma}
\begin{proof}
If we assume that $k_{2}$ and $k_{3}$ are not coprime, then they share a
common factor $b>1$, and we can write $k_{2}=bk_{2}^{\prime }$ and $%
k_{3}=bk_{3}^{\prime }$, where now
\begin{equation}
k_{1}=b\left( q_{1}k_{2}^{\prime }+k_{3}^{\prime }\right) \equiv
bk_{1}^{\prime } \ , \label{B2}
\end{equation}
therefore one gets for the pair of numbers $\left( k_{1},k_{2}\right) = 
\left( bk_{1}^{\prime },bk_{2}^{\prime }\right)$, 
thus $k_{1}$ and $k_{2}$ are not coprime, which contradicts our
initial assumption, therefore $k_{2}$ and $k_{3}$ are necessarily coprime, 
i.e. $b = 1$.
\end{proof}
Here we give a proof of the uniqueness of the sequence (\ref{B4}) of two
arbitrary coprime integers, $k_{1}$ and $k_{2}$, $k_{1}>$ $k_{2}$ and $r>2$. 
Namely, any other pair of coprime numbers leads to a different sequence. 
\begin{thm} \label{theor3}
For a pair of coprime numbers ($k_{1}$,$k_{2}$) the sequence of the $\left(
r-1\right) $-tuple $\left( n_{2},n_{3},...n_{r}\right) $ as given in (\ref%
{B4}) is unique for $r>2$. Two pairs of coprime numbers, ($k_{1}$,$k_{2})$
and ($k_{1}^{\prime }$,$k_{2}^{\prime }$), cannot lead to the same sequence
of integers $\left( n_{2},n_{3},n_{4},...,n_{r}\right)$, for $k_{3}^{\prime
}=k_{3}$ and $k_{4}^{\prime }=k_{4}$, (according to the sequence (\ref{B4})
) .
\end{thm}

\begin{proof}
We construct another pair of coprime numbers ($k_{1}^{\prime }$,$%
k_{2}^{\prime }$) as $k_{1}^{\prime }=m_{1}k_{1}+t_{1}$ and $k_{2}^{\prime
}=m_{2}k_{2}+t_{2}$ (a dilation and a shift for $k_{1}$ and $k_{2}$) with $%
m_{1}$, $m_{2}$ chosen as positive integers and $t_{1}$, $t_{2}$ chosen as
non-negative integers, and $k_{1}^{\prime }>$ $k_{2}^{\prime }$. As $%
k_{1}=n_{2}k_{2}+k_{3}$ we can write
\begin{equation}
k_{1}^{\prime }=n_{2}^{\prime }k_{2}^{\prime }+k_{3}^{\prime }\quad \mathrm{%
\Longrightarrow \quad }m_{1}k_{1}+t_{1}=n_{2}^{\prime }\left(
m_{2}k_{2}+t_{2}\right) +k_{3}^{\prime },  \label{d1}
\end{equation}
or
\begin{equation}
m_{1}\left( n_{2}k_{2}+k_{3}\right) +t_{1}=n_{2}^{\prime }\left(
m_{2}k_{2}+t_{2}\right) +k_{3}^{\prime }.  \label{d2}
\end{equation}
Imposing $k_{3}=k_{3}^{\prime }$ we get the relation
\begin{equation}
n_{2}^{\prime }=\frac{m_{1}k_{2}}{m_{2}k_{2}+t_{2}}n_{2}+\frac{\left(
m_{1}-1\right) k_{3}+t_{1}}{m_{2}k_{2}+t_{2}}  \label{d3}
\end{equation}
and setting $n_{2}^{\prime }=n_{2}$ we obtain a specific relation on $%
m_{1}$, $m_{2}$, $t_{1}$ and $t_{2}$, 
\begin{equation}
n_{2}=\frac{\left( m_{1}-1\right) k_{3}+t_{1}}{\left( m_{2}-m_{1}\right)
k_{2}+t_{2}},  \label{d4}
\end{equation}
that can be a positive integer, due to the arbitrariness of $m_{1}$, $m_{2}$, %
$t_{1}$ and $t_{2}$. The second step of the demonstration consists in
considering the already imposed conditions $k_{3}=k_{3}^{\prime }$
and $n_{2}^{\prime }=n_{2}$, and use them in the second row in the sequence (%
\ref{B4}), 
\begin{eqnarray}
k_{2} =n_{3}k_{3}+k_{4},\; &\textrm{and \ }& k_{2}^{\prime }=n_{3}^{\prime
}k_{3}^{\prime }+k_{4}^{\prime }  \notag \\
&\Longrightarrow &m_{2}k_{2}+t_{2}=n_{3}^{\prime }k_{3}+k_{4}^{\prime
} \\ \notag
&\Longrightarrow & m_{2}\left( n_{3}k_{3}+k_{4}\right) +t_{2}=n_{3}^{\prime
}k_{3}+k_{4}^{\prime }.  
\label{d5}
\end{eqnarray}
For the additional condition $k_{4}=k_{4}^{\prime }$, such that 
\begin{equation}
m_{2}\left( n_{3}k_{3}+k_{4}\right) +t_{2}=n_{3}^{\prime }k_{3}+k_{4}
\end{equation}
we obtain
\begin{equation}
n_{3}^{\prime }=m_{2}n_{3}+\frac{\left( m_{2}-1\right) k_{4}+t_{2}}{k_{3}},
\label{d7}
\end{equation}
and initially considering $m_{2}=1$ we get the relation
\begin{equation}
n_{3}^{\prime }=n_{3}+\frac{t_{2}}{k_{3}}.  \label{d8}
\end{equation}
As so, we may have $n_{3}^{\prime }=n_{3}$ as long as $t_{2}=0$ and 
$m_{2}=1$. Introducing these conditions in Eq. (\ref{d4}), it becomes 
\begin{equation}
n_{2}=-\frac{\left( m_{1}-1\right) k_{3}+t_{1}}{k_{2}\left( m_{1}-1\right) }%
=-\left( \frac{k_{3}}{k_{2}}+\frac{t_{1}}{k_{2}\left( m_{1}-1\right) }%
\right) ,  \label{d9}
\end{equation}
which is a negative number! Thus $n_{3}^{\prime }=n_{3}$ is incompatible with $%
n_{2}^{\prime }=n_{2}$ positive. Now we consider the case $m_{2}>1$:
imposing $n_{3}^{\prime }=n_{3}$ in Eq. (\ref{d7}) we have
\begin{equation}
n_{3}=-\left( \frac{k_{4}}{k_{3}}+\frac{t_{2}}{k_{3}\left( m_{2}-1\right) }%
\right)  \label{d10}
\end{equation}
which is a negative number! Thus, although $n_{2}^{\prime }=n_{2}$ is
possible for $m_{2}>1$, $n_{3}^{\prime }=n_{3}$ is not, so, here too,
necessarily $n_{3}^{\prime }\neq n_{3}$. Summarizing, for two pairs of
coprime numbers ($k_{1}$,$k_{2}$) and ($k_{1}^{\prime }$,$k_{2}^{\prime }$)
that differ under a dilation and shift transformation, we can have both
conditions $k_{3}=k_{3}^{\prime }$ and $k_{4}=k_{4}^{\prime }$ but not $%
n_{2}^{\prime }=n_{2}$, (see Eq. (\ref{d4})), and $n_{3}^{\prime }=n_{3}$,
see (Eq. (\ref{d8})), fulfilled simultaneously. The equalities $%
k_{4}^{\prime }=k_{4}$ and $n_{3}^{\prime }=n_{3}$ happen iff $m_{2}=1$ and $%
t_{2}=0$, but $n_{2}=n_{2}^{\prime }$ is possible only for a negative
integer; when $m_{2}>1$ it is $n_{3}^{\prime }=n_{3}$ that becomes a
negative number. Thence the conditions $k_{3}=k_{3}^{\prime }$, $%
n_{2}^{\prime }=n_{2}$, $k_{4}^{\prime }=k_{4}$ and $n_{3}^{\prime }=n_{3}$
cannot be fulfilled simultaneously.
\end{proof}

\begin{expl} \label{ex10}
For instance, for $k_{1}=17$ and $k_{2}=11$, $n_{2}=1$ and $k_{3}=6$. Now we
choose $t_{1}=6$ and $t_{2}=4$, such that $k_{1}^{\prime
}=23,\,k_{2}^{\prime }=15$ with the condition $n_{2}^{\prime }=n_{2}=1$, we
will get 
\begin{equation}
k_{3}^{\prime }=t_{1}+\frac{\left( k_{2}+t_{2}\right) k_{3}-t_{2}k_{1}}{k_{2}%
}=8,  \label{b67}
\end{equation}
therefore $k_{3}^{\prime }\neq k_{3}$.
\end{expl}
\end{appendix}
%
\begin{acknowledgments}
\noindent SSM thanks the CNPq, a Federal Brazilian Agency, for financial support.
\newline
\noindent We also thank Dr. Gustavo Rigolin, from DF-UFSCar and Prof. Francisco 
C. Alcaraz, from IFSC-USP, for valuable discussions.
\end{acknowledgments}
%


\begin{thebibliography}{9}
%
\bibitem{hill} Hill L. S. \emph{The American Mathematical Monthly} \textbf{38}, No. 3, 135 (1931).

\bibitem{ayan1} Mahalanobis A., \emph{Intern. Math. Forum} \textbf{8}, No. 39, 1939 (2013).

\bibitem{higham1} Higham N.J. and  Lijing Lin, \emph{Linear Algebra and its Applications},  
\textbf{435}, 448 (2011).

\bibitem{diffie}  Diffie W. and Hellman M.E.,  \emph{IEEE Transactions on Information Theory} 
\textbf{IT22}, 644 (1976).

\bibitem{RSA} Rivest R.L., Shamir A., and Adleman L., \emph{Comm. Assoc. Comput. Mach.
} \textbf{21}, 120 (1978).

\bibitem{BB84} Bennett C.H. and Brassard G., \emph{Proceedings of IEEE 
International Conference on Computers, Systems and Signal Processing} \textbf{175} 8, New York (1984).

\bibitem{copper} Coppersmith D. and Winograd S., \emph{J. Symbolic Computation} \textbf{9}, 251 (1990)

\bibitem{knuth1} Knuth D.E., \emph{The Art of Computer Programming}, Vol. II, p. 229, 
3th Edition, Addison Wesley, 1997.

\bibitem{knuth2} idem p. 356. 

\bibitem{khinchin} Khinchin A. Ya., \emph{Continued Fractions}, Dover Publications, 
Mineola, N.Y., 1992.

\bibitem{singh} According to Simon Singh in his page 
https://simonsingh.net/media/articles/maths-and-science/unsung-heroes-of-cryptography/

\bibitem{adrian} Adrian D. et al., \emph{Imperfect Forward Secrecy: 
How Diffie-Hellman Fails in Practice}, in the {22nd ACM Conference on Computer 
and Communications Security, October 12-16, 2015}, Denver, Colorado, US, 
DOI: http://dx.doi.org/10.1145/2810103.2813707

\bibitem{cohn} Cohn H., \emph{Advanced Number Theory}, Dover Publications Inc., 
New York, 1980.


%
\end{thebibliography}
\end{document}